\documentclass[11pt, a4paper]{article}

\usepackage[margin=1in]{geometry}

\usepackage{hyperref}

 \usepackage{amsfonts}
 \usepackage{amsmath}
 \usepackage{amssymb}
 \usepackage{amsthm}
\usepackage{bbm} 
\usepackage{cite}
\usepackage{color}
\usepackage{graphicx}
\usepackage{mathrsfs} 
\usepackage[normalem]{ulem}

\allowdisplaybreaks

\def\extraspacing{\vspace{3mm} \noindent}
\def\figcapup{\vspace{-1mm}}
\def\figcapdown{\vspace{-0mm}}

\def\vgap{\vspace{1mm}}

\def\tabpos{\hspace{4mm} \= \hspace{4mm} \= \hspace{4mm} \= \hspace{4mm} \= \hspace{4mm} \= \hspace{4mm} \= \hspace{4mm} \= \hspace{4mm} \= \hspace{4mm} \= \hspace{4mm} \= \hspace{4mm} \= \hspace{4mm} \= \hspace{4mm} \= \hspace{4mm} 
\kill}
\newcommand{\mytab}[1]{\begin{tabbing}\tabpos #1\end{tabbing}}

 \theoremstyle{plain}
 \newtheorem{theorem}{Theorem}[section]
 \newtheorem{proposition}[theorem]{Proposition}
 \newtheorem{lemma}[theorem]{Lemma}
 
 \theoremstyle{definition}

 \theoremstyle{remark}
 
\newtheorem{fact}{Fact}[section]

\newcommand{\boxminipg}[2]{\begin{center}\fbox{\begin{minipage}{#1}#2\end{minipage}}\end{center}}
\newcommand{\minipg}[2]{\begin{center}\begin{minipage}{#1}#2\end{minipage}\end{center}}
\newcommand{\myitems}[1]{\begin{itemize} #1 \end{itemize}}
\newcommand{\myenums}[1]{\begin{enumerate} #1 \end{enumerate}}

\newcommand{\bm}[1]{\textrm{\boldmath${#1}$}}

\newcommand{\myeqn}[1]{\begin{eqnarray}#1\end{eqnarray}}
\newcommand{\set}[1]{\{#1\}}

\newcommand{\explain}[1]{(\textrm{#1})}

\def\mit{\mathit}

\def\eps{\epsilon}
\def\fr{\frac}
\def\-{\mbox{-}}

\def\real{\mathbb{R}}

\newcommand{\ceil}[1]{\left \lceil #1 \right \rceil}
\newcommand{\floor}[1]{\left \lfloor #1 \right \rfloor}

\def\nn{\nonumber}
\def\setm{\setminus}

\DeclareMathOperator*{\polylog}{polylog}

\allowdisplaybreaks

\def\vgap{\vspace{1mm}}
\def\vslit{\vspace{0.5mm}}
\def\extraspacing{\vspace{2mm} \noindent}
\def\figcapup{\vspace{-3mm}}
\def\figcapdown{\vspace{-3mm}}

\def\Zdom{\mathbb{Z}}

\def\C{\mathscr{C}}
\def\D{\mathcal{D}}
\def\M{\mathcal{M}}
\def\T{\mathcal{T}}

\def\anc{\text{anc}}
\def\del{\text{del}}
\def\diam{\mit{diam}}
\def\first{\text{1st}}
\def\geo{\text{geo}}
\def\last{\text{last}}
\def\level{\mit{level}}
\def\net{\text{net}}
\def\qry{\text{qry}}
\def\out{\text{out}}
\def\start{\text{start}}
\def\upd{\text{upd}}

\title{Proximity Graphs for Similarity Search: Fast Construction, Lower Bounds, and Euclidean Separation}

\author{
	Shangqi Lu \\[2mm]
	HKUST-Guangzhou \\
	{\em shangqilu@hkust-gz.edu.cn}
	\and Yufei Tao \\[2mm]
	CUHK \\
	{\em taoyf@cse.cuhk.edu.hk}
}

\begin{document}

\maketitle

\begin{abstract}
	Proximity graph-based methods have emerged as a leading paradigm for approximate nearest neighbor (ANN) search in the system community. This paper presents fresh insights into the theoretical foundation of these methods. We describe an algorithm to build a proximity graph for $(1+\eps)$-ANN search that has $O((1/\eps)^\lambda \cdot n \log \Delta)$ edges and guarantees $(1/\eps)^\lambda \cdot \polylog \Delta$ query time. Here, $n$ and $\Delta$ are the size and aspect ratio of the data input, respectively, and $\lambda = O(1)$ is the doubling dimension of the underlying metric space. Our construction time is near-linear to $n$, improving the $\Omega(n^2)$ bounds of all previous constructions. We complement our algorithm with lower bounds revealing an inherent limitation of proximity graphs: the number of edges needs to be at least $\Omega((1/\eps)^\lambda \cdot n + n \log \Delta)$ in the worst case, up to a subpolynomial factor. The hard inputs used in our lower-bound arguments are non-geometric, thus prompting the question of whether improvement is possible in the Euclidean space (a key subclass of metric spaces). We provide an affirmative answer by using geometry to reduce the graph size to $O((1/\eps)^\lambda \cdot n)$ while preserving nearly the same query and construction time.
\end{abstract}

\thispagestyle{empty}
\pagebreak

\setcounter{page}{1}

\section{Introduction}\label{sec:intro}

Approximate nearest neighbor (ANN) search is fundamental to similarity retrieval and plays an imperative role in a wide range of database applications, such as recommendation systems, entity matching, multimedia search, DB for AI, and so on. In the past decade, proximity graph-based approaches (e.g., HNSW \cite{my20}) have become a dominant paradigm for ANN search in the system community. Numerous articles \cite{aep25,czk+24,fxwc19,ggx+25,hd16,gl23,mplv14,my20,pcc+23,pjw24,sds+19,wxy+24,wxyw21,zth+23} in venues like SIGMOD, VLDB, NeurIPS, etc.\ have demonstrated proximity graphs' superior empirical performance on real-world data, even against methods with solid worst-case guarantees. Despite their empirical success, however, the theoretical underpinnings of proximity graphs remain largely unexplored. This raises a critical question:

\vgap

\minipg{0.8\linewidth}{
	Is the performance of proximity graphs driven by specific properties of the datasets evaluated, or do they possess inherent theoretical strengths?
}

\vgap

\noindent Given the vast popularity of proximity graphs, we believe that there is an urgent need to deepen our understanding of their combinatorial nature, thereby enabling us to analyze and predict their efficacy across diverse contexts.

\subsection{Problem Definitions} \label{sec:intro:prob}

We consider a metric space $(\M, D)$ where
\myitems{
	\item $\M$ is a (possibly infinite) set where each element is called a {\em point};
	\item $D$ is a function that, given two points $p_1, p_2 \in \M$, computes in constant time a non-negative real value as their {\em distance}, denoted as $D(p_1, p_2)$.
}
The function $D$ satisfies (i) identity of indiscernibles: $D(p_1, p_2) = 0$ if and only if $p_1 = p_2$, (ii) symmetry: $D(p_1, p_2) = D(p_2, p_1)$, and (iii) triangle inequality: $D(p_1, p_2) \le D(p_1, p_3) + D(p_2, p_3)$.

\vgap

Let $P$ be a set of $n \ge 2$ points from $\M$, which we refer to as the {\em data points}. Given a point $q \in \M$, a point $p^* \in P$ is a {\em nearest neighbor} (NN) of $q$ if $D(p^*, q) \le D(p, q)$ holds for all $p \in P$. For a value $\eps \in (0, 1]$, a point $p \in P$ is called a {\em $(1+\eps)$-approximate nearest neighbor} of $q$ if $D(p, q)$ $\le$ $(1+\eps) \cdot D(p^*, q)$.

\vgap

Consider a simple directed graph $G$, where each point of $P$ corresponds to a vertex in $G$, and vice versa. We call $G$ a {\em $(1+\eps)$-proximity graph} (PG) if, given any query point $q \in \M$ and any data point $p_\start \in P$, the following procedure always returns a $(1+\eps)$-ANN of $q$:

\mytab{
	\> {\bf greedy}$(p_\start, q)$ \\
	\> 1. \> $p^\circ \leftarrow p_\start$ \hspace{5mm} /* the first hop */ \\
	\> 2. \> {\bf repeat} \\
	\> 3. \>\> $p^+_\out \leftarrow$ the out-neighbor of $p^\circ$ closest to $q$
	\hspace{5mm} /* $p^+_\out =$ nil if $p^\circ$ has no out-neighbors */ \\
	\> 4. \>\> {\bf if} $p^+_\out =$ nil {\bf or} $D(p^\circ, q) \le D(p^+_\out, q)$ {\bf then return $p^\circ$} \\
	\> 5. \>\> $p^\circ \leftarrow p^+_\out$ \hspace{5mm} /* the next hop */
}

\noindent At each $p^\circ$ --- henceforth referred to as a {\em hop vertex} --- the procedure computes $D(p_\out, q)$ for every out-neighbor $p_\out$ of $p^\circ$. The sequence of hop vertices (a.k.a.\ data points)  visited have strictly descending distances (i.e., $D(p^\circ, q)$) to $q$.

\vgap


Although {\bf greedy} always returns a correct answer, it can be slow because it may need to visit a long senquence of vertices before termination. However, a good $(1+\eps)$-PG should allow {\bf greedy} to find an $(1+\eps)$-ANN after only a small number of distance computations. Formally, we say that $G$ ensures {\em query time} $Q$ if the following algorithm always returns a $(1+\eps)$-ANN of $q$, regardless of the choice of $p_\start$ and $q$:

\mytab{
	\> {\bf query}$(p_\start, q, Q)$ \\
	\> 1. \> run {\bf greedy}$(p_\start, q)$ until it self-terminates or has computed $Q$ distances \\
	\> 2. \> {\bf if} self-termination {\bf then return} the output of {\bf greedy} \\
	\> 3. \> {\bf else} {\bf return} the last hop vertex $p^\circ$ visited by {\bf greedy}
}

\noindent Note that a ``$Q$ query time'' guarantee defined as above directly translates into a maximum running time of $O(Q)$ because distance calculation is the bottleneck of {\bf greedy}.

\vgap

Proximity graphs definitely exist: the complete graph $G$ --- namely, there is an edge from a data point to every other data point --- is a proximity graph for any $\eps > 0$. However, this $G$ has $\Theta(n^2)$ edges and can only ensure a query time of $\Omega(n)$. Research on proximity graphs revolves around two questions:
\myitems{
	\item {{\bf Q1:}} How to build a smaller proximity graph with faster  query time?

	\item {{\bf Q2:}} What are the limitations of proximity graphs?
}
Besides $\eps$ and $n$, we will describe our results using two other parameters, as introduced next.

\extraspacing {\bf Aspect Ratio.} In general, for any subset $X \subseteq \M$, its {\em diameter} --- denoted as $\diam(X)$ --- is the maximum distance of two points in $X$, while its {\em aspect ratio} is the ratio between $\diam(X)$ and the smallest inter-point distance in $X$. The first extra parameter we adopt is the aspect ratio of $P$, denoted as $\Delta$.

\extraspacing {\bf Doubling Dimension.} For any point $q \in \M$ and any real value $r \ge 0$, define $B(q, r)$ --- referred to as a {\em ball} with radius $r$ --- as the set $\set{p \in \M \mid D(p, q) \le r}$. The second extra parameter we adopt is the {\em doubling dimension} of the metric space $(\M, D)$. This is the smallest value $\lambda$ satisfying the following condition: for any $r > 0$, every ball of radius $r$ can be covered by the union of at most $2^\lambda$ balls of radius $r/2$. The value $\lambda$ measures  the ``intrinsic dimensionality'' of a metric space. Our discussion throughout the paper will assume $\lambda$ to be bounded by a constant.


\extraspacing {\bf Mathematical Conventions.} For an integer $x \ge 1$, the notation $[x]$ represents the set $\set{1, 2, ..., x}$. If $p$ is a point in $\real^d$, its $i$-th coordinate is denoted as $p[i]$ for each $i \in [d]$. Given two points $p, q \in \real^d$, we use $L_2(p, q)$ and $L_\infty(p, q)$ to represent their distance under the $L_2$ and $L_\infty$ norms, respectively. All angles are measured in Radians, and all logarithms have base 2. In a directed graph, we use the notation $(u, v)$ to represent a directed edge from vertex $u$ to vertex $v$. By saying that a random event occurs ``with high probability'' (w.h.p.\ for short), we mean that the event happens with probability at least $1-1/n^c$ where $c$ can be set to an arbitrarily large constant.


\subsection{Previous Work} \label{sec:intro:prev}

Many methods have been proposed in the system community for proximity graph construction, with small-world graph \cite{mplv14}, DiskANN \cite{sds+19}, NSG \cite{fxwc19}, and HNSW \cite{my20} being notable examples. Although the efficiency of these methods has been demonstrated through extensive experiments, limited research has focused on understanding their theoretical characteristics. In \cite{ix23}, Indyk and Xu addressed the issue by conducting a comprehensive analysis of the worst-case performance of the existing PG-based methods. They found that DiskANN is the only method that enjoys non-trivial guarantees. Specifically, DiskANN builds a $(1+\eps)$-PG in $O(n^3)$ time that has $O((1/\eps)^{\lambda} \cdot n \log \Delta)$ edges and guarantees a query time of $O((1/\eps)^{\lambda} \cdot \log^2 \Delta)$.

\vgap

Diwan et al.\ \cite{dgm+24} considered the special scenario where $P = \M$ (or equivalently, every query point originates from $P$). In that case, they proved the existence of a $(1+\eps)$-PG that has $O(n^{1.5} \log n)$ edges and ensures $O(\sqrt{n} \log n)$ query time. Given any $q \in P$, their PG allows the greedy algorithm of Section~\ref{sec:intro:prob} to find an exact NN of $q$ (which is $q$ itself) within the aforementioned time bound; hence, their construction works for any value of $\eps$.

\vgap

An important class of metric spaces is the Euclidean space $\real^d$ coupled with the $L_2$ norm, for which several authors have studied how to leverage geometry to build PGs, assuming the dimensionality $d$ to be a constant. Specifically, the construction by Arya and Mount \cite{am93} produces a $(1+\eps)$-PG with $O((1/\eps)^d \cdot n)$ edges but $\Omega(n)$ query time.  Clarkson \cite{c94} presented another construction that yields a $(1+\eps)$-PG with $O((1/\eps)^{(d-1)/2} \cdot n \log (\Delta/\eps))$ edges and $O((1/\eps)^{(d+1)/2} \cdot \log (\Delta/\eps) \cdot \log \Delta)$ query time. The construction of \cite{am93} and \cite{c94} takes $O((1/\eps)^d \cdot n^2)$ expected time and $O((1/\eps)^{d-1} \cdot n^2 \log (\Delta/\eps))$ time, respectively. As the doubling dimension $\lambda$ of $\real^d$ satisfies $d \le \lambda = O(d)$, Clarkson's bounds are better than those of DiskANN on $(\real^d, L_2)$.


\vgap

We emphasize that the purpose of this work is to explore the theory of proximity graphs --- in particular, to seek answers for {\bf Q1} and {\bf Q2} in Section~\ref{sec:intro:prob} --- rather than designing new data structures for ANN search. Readers interested in {\bf non}-PG-based ANN structures with strong performance guarantees may refer to the representative works \cite{amn+98, amm09, ar15, cg06, diim04, hm06, im98, kl04} and the references therein.

\subsection{Our Results} \label{sec:intro:ours}

We present new answers to both questions {\bf Q1} and {\bf Q2}. Our first main result is an algorithm for building proximity graphs:

\begin{theorem} \label{thm:ub-metric}
	For a metric space with a constant doubling dimension $\lambda$, there is a $(1+\eps)$-PG that has $O((1/\eps)^\lambda \cdot n \log \Delta)$ edges and $O((1/\eps)^{\lambda} \cdot \log^2 \Delta)$ query time, where $n$ and $\Delta$ are the size and aspect ratio of the data input, respectively. We can construct such a graph in $(1/\eps)^{\lambda} \cdot n  \polylog (n \Delta)$ time.
\end{theorem}

Our algorithm, which improves DiskANN, is the first in the literature whose construction cost avoids a quadratic dependence on $n$. We then continue to explore whether the graph size in Theorem~\ref{thm:ub-metric} can be significantly reduced, in particular:
\myitems{
	\item {{\bf Q2.1:}} For constant $\eps$, both DiskANN and our solution produce a graph of $O(n \log \Delta)$ edges. Is the $\log \Delta$ factor an artifact? In other words, for constant $\eps$, is there a $(1+\eps)$-PG of $O(n)$ edges? What if the query time is allowed to be arbitrarily large?

	\vgap

	\item {{\bf Q2.2:}} For non-constant $\eps$, both DiskANN and our solution have the term $(1/\eps)^\lambda \cdot n$ in the graph size. Is the $(1/\eps)^\lambda$ factor necessary? Again, what if the query time is allowed to be arbitrarily large?

}

Our second theorem gives lower bounds justifying both the $\log \Delta$ and the $(1/\eps)^\lambda$ factors.

\begin{theorem} \label{thm:lbs}
	The following statements are true:
	\myenums{
		\item For any integers $\Delta$ and $n$ that are powers of 2 satisfying $n \ge 2$ and $n^2 \le 2\Delta \le 2^n$, there is a set $P$ of $\Theta(n)$ points with aspect ratio $\Delta$ from a metric space whose doubling dimension is 1, such that any 2-PG for $P$ must have $\Omega(n \log \Delta)$ edges, regardless of the query time allowed.

		\vgap

		\item Given any integers $s \ge 2$, $t \ge 1$, and constant $d \ge 1$, there is a set $P$ of $n = s^d \cdot t$ points with aspect ratio $\Delta = O(n)$ from a metric space whose doubling dimension $\lambda$ is at most $\log (1+2^d)$ such that, for $\eps = 1/(2s)$, any $(1+\eps)$-PG for $P$ must have $\Omega(s^d \cdot n)$ edges, regardless of the query time allowed.
	}
\end{theorem}

Statement (1) of Theorem~\ref{thm:lbs} answers {\bf Q2.1} in a straightforward manner (it is worth mentioning that the theorem still holds even when the constant 2 in ``2-PG'' is replaced with any other constant greater than 1). To see the connection between Statement (2) and {\bf Q2.2}, note that the difference between $\lambda$ and $d$ converges to 0 when $d$ increases. The term $s^d \cdot n$ --- which is $(\fr{1}{2\eps})^d \cdot n$ --- is greater than $(\fr{1}{2 \eps})^{\lambda - \delta} \cdot n$ for any arbitrarily small constant $\delta > 0$ when $d$ is sufficiently large. Thus, an $(1+\eps)$-PG must have $\Omega((1/\eps)^\lambda \cdot n)$ edges up to a subpolynomial factor. This also provides a justification on the construction time in Theorem~\ref{thm:ub-metric}, which can no longer be improved by more than a sub-polynomial factor. It is worth mentioning that the parameter $t$ in Statement (2) permits the lower bound to hold for a wide range of $\eps$.

\vgap

The astute reader may have noticed from Statement (2) that, when $\eps = O(1/n^{1/\lambda})$, in the worst case every $(1+\eps)$-PG must have $\Omega(n^2)$ edges under our construction of $P$, essentially the worst possible! This does not contradict the result of \cite{dgm+24} --- which as mentioned before argues for the existence of a $(1+\eps)$-PG of size $O(n^{1.5} \log n)$ for any $\eps$ --- because the result of \cite{dgm+24} holds only in the (very) special case where $P = \M$.


\vgap

The hard instances utilized to establish Theorem~\ref{thm:lbs} are non-geometric. This prompts another intriguing question: does the Euclidean space allow the fast construction of a $(1+\eps)$-PG of $O((1/\eps)^\lambda \cdot n)$ edges and $(1/\eps)^\lambda \cdot \polylog (n\Delta)$ query time? Our last main result answers the question in the affirmative:

\begin{theorem} \label{thm:ub-geo}
	Given any set $P$ of $n$ points in the metric space $(\real^d, L_2)$ where $d = O(1)$, there is a $(1+\eps)$-PG that has $O((1/\eps)^{\lambda} \cdot n)$ edges and ensures a query time of $O((1/\eps)^\lambda \cdot \log^2 \Delta + (1/\eps)^{d-1} \log n \cdot \log^2 \Delta)$, where $n$ and $\Delta$ are the size and aspect ratio of the data input, respectively, and $\lambda = O(d)$ is the doubling dimension of $(\real^d, L_2)$. W.h.p., we can build such a graph in $(1/\eps)^{\lambda} \cdot n  \polylog (n \Delta)$ time.
\end{theorem}

For constant $\eps$ (an important use case in practice), our PG is the first in the literature that has $O(n)$ edges and guarantees $\polylog (n \Delta)$ query time, not to mention that it is also the first PG that can be constructed in $n \polylog (n \Delta)$ time. Its size bound draws a separation between the Euclidean space and general metric spaces (as per Statement (1) of Theorem~\ref{thm:lbs}).


\extraspacing {\bf A Paradigm Critique.} Our results enable an objective critique on the proximity-graph paradigm, at least in the regime where the doubling dimension is small. On the bright side, Theorem~\ref{thm:ub-metric} shows that it {\em is} possible to build a good PG in time near-linear to $n$ (expensive construction has been a major issue in the paradigm's literature). However, our hardness results in Theorem~\ref{thm:lbs} clearly indicate that {\em space} is an inherent defect of the paradigm. In particular, one should abandon the hope to attain a clean space complexity of $O(n)$, except in the restricted scenario where both $\eps$ and $\Delta$ are constants. In contrast, the theory field has already discovered \cite{cg06,hm06} a data structure of $O(n)$ space that can answer a $(1+\eps)$-ANN query in $O(\log n) + (1/\eps)^{O(\lambda)}$ time, regardless of $\Delta$. Our lower bounds, however, do not rule out a $(1+\eps)$-PG of $O((1/\eps)^\lambda \cdot n + n \log \Delta)$ edges. Finding a way to meet this bound or arguing against its possibility would make an interesting intellectual challenge.

\vgap

We emphasize that it is not our objective to dismiss PGs as an inferior paradigm. The constituting concepts of the paradigm are elegant, especially  the convenient flexibility in choosing the $p_\start$ point for {\bf greedy}, which suggests that the paradigm may have strengths in enforcing load-balancing in network-scale distributed computing (found in ``Internet-of-Things'' applications).

\section{A Proximity Graph with Fast Construction} \label{sec:alg1}

This section serves as a proof of Theorem~\ref{thm:ub-metric}. The key of our proof is to explain how ``$r$-nets'' --- a tool from computational geometry as defined below --- are useful for building proximity graphs:

\vgap

\minipg{0.85\linewidth}{
	Given a subset $X \subseteq \M$ and a value $r > 0$, an {\em $r$-net} of $X$ is a subset $Y \subseteq X$ satisfying:
	\myitems{
		\item {(separation property)} $D(y_1, y_2) \ge r$ for any two distinct points $y_1, y_2 \in Y$;
		\item {(covering property)}
		$X \subseteq \bigcup_{y \in Y} B(y, r)$, i.e., for $\forall x \in X$, $\exists$ a point $y \in Y$ with $D(x, y) \le r$.
	}
}

\vslit

\noindent The rest of the section is organized as follows. We will first define the proposed proximity graph in Section~\ref{sec:alg1:graph} and then analyze its properties, size, and query time in Sections~\ref{sec:alg1:pg-proof} and \ref{sec:alg1:space-qry}. Finally, Section~\ref{sec:alg1:construction} will explain how to construct the graph efficiently.

\subsection{The Graph} \label{sec:alg1:graph}

We consider that the smallest inter-point distance in $P$ is 2 (as can be achieved by scaling $D$ appropriately). In other words, the aspect ratio of $P$ is $\Delta = \diam(P) / 2$.
Define
\myeqn{
	h = \ceil{\log \diam(P)}. \label{eqn:alg1:h}
}
For each $i \in [0, h]$, define
\myeqn{
	Y_i = \text{a $2^i$-net of $P$}. \label{eqn:alg1:Y_i}
}
Note that $Y_0$ must be $P$ (as the smallest inter-point distance is 2). Furthermore, define
\myeqn{
	\eta &=& \ceil{\log(1 + 2/\eps)} \label{eqn:alg1:eta} \\
	\phi &=& 1+2^{\eta+1}. \label{eqn:alg1:rho}
}
Clearly, $\eta \ge 2$ and $9 \le \phi = \Theta(1/\eps)$.

\vgap

We now formulate a graph $G_\net$. Every vertex of $G_\net$ is a point in $P$ and vice versa. For each $p \in P$, decide its out-edges as follows:
\boxminipg{0.9\linewidth}{
	for each $i \in [0, h]$, create an edge $(p, y)$ in $G_\net$ for every $y \in Y_i$ satisfying $D(p, y) \le \phi \cdot 2^i$.
}

\begin{proposition} \label{prop:alg1:out-deg}
	Every vertex (a.k.a.\ point) in $G_\net$ has an out-degree at least 1.
\end{proposition}

\begin{proof}
	Fix any point $p \in P$. If $p \notin Y_i$ for some $i \in [0, h]$, there is a point $y \in Y_i$ with $D(p, y) \le 2^i$ due to the covering property. This $y$ must be an out-neighbor of $p$. Next, we assume that $p \in Y_i$ for all $i \in [0, h]$.

	\vgap

	Denote by $j$ the highest value of $i$ satisfying $|Y_i| \ge 2$. Let $y$ be any point in $Y_i$ different from $p$. We argue that $y$ must be an out-neighbor of $p$, i.e., $D(p, y) \le \phi \cdot 2^j$. Indeed, this is true if $j = h$ because $D(p, y) \le \diam(P) \le 2^h$. Consider now $j < h$. It follows from the definition of $j$ that $|Y_{j+1}| = 1$, in which case the covering property tells us $D(p, y) \le 2^{j+1} < \phi \cdot 2^j$.
\end{proof}

\subsection{Properties of $\bm{G_\net}$} \label{sec:alg1:pg-proof}

Let $G$ be a simple directed graph whose vertices have one-one correspondence to the points in $P$. We say that $G$ is {\em $(1+\eps)$-navigable} if the following condition holds for every data point $p \in P$ and every query point $q \in \M$:
\myitems{
	\item either $p$ is a $(1+\eps)$-ANN of $q$,
	\item or $p$ has an out-neighbor $p_\out$ satisfying $D(p_\out, q) < D(p, q)$.
}
The following fact is folklore (see Appendix~\ref{app:alg1:pg-proof:navigable} for a proof):

\begin{fact} \label{fact:alg1:pg-proof:navigable}
    $G$ is a $(1+\eps)$-PG of $P$ if and only if $G$ is $(1+\eps)$-navigable.
\end{fact}

\noindent Next, we show that the graph $G_\net$ built earlier is $(1+\eps)$-navigable and, therefore, a $(1+\eps)$-PG of $P$. In addition, we will establish a {\em log-drop property} of $G_\net$ that is crucial for the technical development in the later parts of the paper.

\vgap


\begin{lemma} \label{lmm:alg1:pg-proof:help2}
	Fix an arbitrary point $q \in \M$, and an arbitrary point $p^\circ \in P$ that is not a $(1+\eps)$-ANN of $q$. Define
    \myeqn{
		p^+_\out = \text{the out-neighbor of $p^\circ$ closest to $q$.}
		\label{eqn:alg1:p+_out}
	}
	Both of the following statements are true:
    \myenums{
		\item $D(p^+_\out, q) < D(p^\circ, q)$.

		\vgap

		\item {\em {\bf [The log-drop property]}} Let $\varrho$ be any point in $P$ satisfying $D(\varrho, q) \le D(p^+_\out, q)$. If $\varrho$ is not a $(1+\eps)$-ANN of $q$, then
		\myeqn{
			\ceil{\log D(\varrho, p^*)} < \ceil{\log D(p^\circ, p^*)}
			\label{eqn:alg1:log-drop}
		}
		where $p^*$ is an exact NN of $q$.
    }
\end{lemma}

Statement (1) of Lemma~\ref{lmm:alg1:pg-proof:help2} indicates that $G_\net$ is $(1+\eps)$-navigable.

\extraspacing {\bf Proof of Lemma~\ref{lmm:alg1:pg-proof:help2}.} Define:
\myeqn{
	\alpha &=& 
	\ceil{\log D(p^\circ, p^*)} \label{eqn:alg1:alpha} \\
	\beta &=& \max\{\alpha - \eta - 1, 0\} \label{eqn:alg1:beta}
}
where $\eta$ is given in \eqref{eqn:alg1:eta}. Note that $\alpha \ge 1$ and $0 \le \beta \le h$, where $h$ is given in \eqref{eqn:alg1:h}. Furthermore, $\beta \le \alpha - 1$. Define:
\myeqn{
	y^\circ = \text{an arbitrary point in $Y_\beta$ such that $D(p^*, y^\circ) \le 2^\beta$}
	\label{eqn:alg1:y-circ}
}
Such $y^\circ$ exists because $Y_\beta$ is a $2^\beta$-net of $P$ (the covering property). Note that $y^\circ \ne p^\circ$ because $D(p^\circ, p^*) > 2^{\alpha-1} \ge 2^\beta \ge D(y^\circ, p^*)$.
Using $\beta \ge \alpha - \eta - 1$ (see \eqref{eqn:alg1:beta}), we can derive
    \myeqn{
		D(p^\circ, y^\circ) \le D(p^\circ, p^*) + D(p^*, y^\circ) \le 2^\alpha + 2^\beta = 2^\beta \cdot (2^{\alpha-\beta} + 1) \le 2^\beta \cdot (2^{\eta + 1} + 1) = \phi \cdot 2^\beta. \nn
    }
    Hence, $y^\circ$ must be an out-neighbor of $p^\circ$ by how $G_\net$ is built.

%
%
		\begin{fact} \label{fact:alg1:pg-proof:help2:help1}
		    If a point $p \in P$ is not a $(1+\eps)$-ANN of $q$, either $\log D(p, p^*) \le \alpha - 1$ or $D(p, q) > D(y^\circ, q)$.
		\end{fact}

	Before delving into the fact's proof, let us note how it implies Statements (1) and (2) of Lemma~\ref{lmm:alg1:pg-proof:help2}:
	\myitems{
		\item Applying the fact with $p = p^\circ$ tells us $D(p^\circ, q) > D(y^\circ, q)$ because \eqref{eqn:alg1:alpha} suggests $\log D(p^\circ, p^*) > \alpha - 1$. This, together with the definition of $p^+_\out$ in \eqref{eqn:alg1:p+_out}, proves Statement (1) of Lemma~\ref{lmm:alg1:pg-proof:help2}.

		\vgap

		\item Applying the fact with $p = \varrho$ proves Statement (2) because $D(\varrho, q) \le D(p^+_\out, q) \le D(y^\circ, q)$.
	}

	\vslit

	\begin{proof} [Proof of Fact~\ref{fact:alg1:pg-proof:help2:help1}]
	Suppose that $D(p, p^*) > 2^{\alpha - 1}$ and $D(p, q) \le D(y^\circ, q)$ hold simultaneously. We argue that in this case $p$ must be a $(1+\eps)$-ANN of $q$, which will then validate Fact~\ref{fact:alg1:pg-proof:help2:help1}.

	\vgap

	Consider first $\beta = 0$. From \eqref{eqn:alg1:y-circ}, we know $D(p^*, y^\circ) \le 1$, implying that $y^\circ = p^*$ (the inter-point distance in $P$ is at least 2). The condition $D(p, q) \le D(y^\circ, q)$ asserts that $p$ must be an exact NN of $q$. The subsequent discussion assumes $\beta = \alpha - \eta - 1 > 0$.

	\vgap

	By $D(p, p^*) > 2^{\alpha-1}$ and $D(y^\circ, p^*) \le 2^\beta = 2^{\alpha-\eta-1}$, we obtain
	\myeqn{
		&& D(y^\circ, p^*) < D(p, p^*) / 2^\eta \label{eqn:alg1:pg-proof:help2:help1:help1} \\
		&&\Rightarrow D(p,q) \le D(y^\circ, q) \le D(y^\circ, p^*) + D(p^*, q) <  D(p, p^*) / 2^\eta + D(p^*, q). \label{eqn:alg1:pg-proof:help2:help1:help2}
	}
	On the other hand, the triangle inequality shows
	\myeqn{
		D(p, p^*) &\le& D(p, q) + D(q, p^*) \nn \\
		\explain{by \eqref{eqn:alg1:pg-proof:help2:help1:help2}}
		&<&
		D(p, p^*) / 2^\eta + 2 \cdot D(p^*, q) \nn
	}
	Subtracting $D(p, p^*) / 2^\eta$ from both sides and rearranging terms gives:
	\myeqn{
		D(p, p^*) < \fr{2^{\eta + 1}}{2^\eta - 1} \cdot D(p^*, q) \nn
	}
	Plugging the above into \eqref{eqn:alg1:pg-proof:help2:help1:help2}, we obtain:
	\myeqn{
		D(p, q) < \fr{2}{2^\eta - 1} D(p^*, q) + D(p^*, q) \le (1+\eps) \cdot D(p^*, q) \nn
	}
	where the last step used $2^\eta - 1 \ge 2/\eps$ (due to \eqref{eqn:alg1:eta}). Hence, $p$ is a $(1+\eps)$-ANN of $q$.
\end{proof}


\subsection{Size and Query Time} \label{sec:alg1:space-qry}

The following is a rudimentary fact of metric spaces:

\begin{fact} \label{fact:alg1:pg-proof:aspect}
    Consider any subset $X$ $\subseteq$ $\M$.
    If $X$ has aspect ratio $A$, then $|X| = O(A^\lambda)$.
\end{fact}

See Appendix~\ref{app:alg1:pg-proof:aspect} for a proof; recall that $\lambda$ is the doubling dimension of the metric space $(\M, D)$. Fact~\ref{fact:alg1:pg-proof:aspect} assures us that every vertex in $G_\net$ has an out-degree of $O(\phi^\lambda \cdot \log \Delta)$, where $\phi = \Theta(1/\eps)$ is given in \eqref{eqn:alg1:rho}. To see why, consider any point $p \in P$ and any $i \in [0, h]$. Set $X$ to the set of points in $Y_i$ that are out-neighbors of $p$. Recall that $p$ has an edge to $y \in Y_i$ only if $D(p, y) \le \phi \cdot 2^i$. Hence, $X$ is a subset of the ball $B(p, \phi \cdot 2^i)$, implying that $\diam(X) \le 2 \phi \cdot 2^i$. On the other hand, by definition of $2^i$-net, the distance between two (distinct) points in $X$ is at least $2^i$ (the separation property). Thus, the aspect ratio of $X$ is at most $2\phi$, which by Fact~\ref{fact:alg1:pg-proof:aspect} yields $|X| = O((2 \phi)^\lambda) = O(\phi^\lambda)$ because $\lambda = O(1)$. As $i$ has $h+1 = O(\log \Delta)$ choices, $p$ can have at most $O(\phi^\lambda \cdot \log \Delta)$ out-edges. The total number of edges in $G_\net$ is therefore $O((1/\eps)^\lambda \cdot n \log \Delta)$.

\vgap

Next, we prove that $G_\net$ guarantees a query time of $O(\phi^\lambda \cdot \log^2 \Delta)$. Recall that each {\em iteration} of {\bf greedy} --- Lines 3-5 of its pseudocode in Section~\ref{sec:intro:prob} --- visits a new hop vertex $p^\circ$. Once an iteration encounters a $p^\circ$ that is a $(1+\eps)$-ANN of $q$, it will definitely return a $(1+\eps)$-ANN of $q$ because the hop vertices encountered in the subsequent iterations can only be closer to $q$.

\vgap

We argue that, after at most $h$ iterations, the current hop vertex must be a $(1+\eps)$-ANN of $q$. Our weapon is Statement (2) of Lemma~\ref{lmm:alg1:pg-proof:help2}. Suppose that the hop vertex $p^\circ$ of a certain iteration is not a $(1+\eps)$-ANN of $q$. The next iteration of {\bf greedy} will hop to the vertex $p^+_\out$ defined in \eqref{eqn:alg1:p+_out}. If $p^+_\out$ is not a $(1+\eps)$-ANN either, setting $\varrho = p^+_\out$ in Statement (2) of Lemma~\ref{lmm:alg1:pg-proof:help2} yields:
\myeqn{
	\ceil{\log D(p^+_\out, p^*)} < \ceil{\log D(p^\circ, p^*)}.
	\label{eqn:alg1:space-qry:drop}
}
This means that the value of $\ceil{\log D(p^\circ, p^*)}$ at the beginning of an iteration must decrease by at least 1 compared to the previous iteration. This can happen at most $h$ times before $D(p^\circ, p^*)$ drops below 2, at which moment we must have $p^\circ = p^*$.

\vgap

Each iteration calculates $O(\phi^\lambda \cdot \log \Delta)$ distances because, as proved earlier, each point has an out-degree of $O(\phi^\lambda \cdot \log \Delta)$. The total query time is thus $O(\phi^\lambda \cdot \log^2 \Delta) = O((1/\eps)^\lambda \cdot \log^2 \Delta)$.

\subsection{Construction} \label{sec:alg1:construction}

A primary benefit of connecting proximity graphs to $r$-nets is that we can leverage the rich algorithmic literature of $r$-nets to construct $G_\net$ efficiently. Consider the following procedure:

\mytab{
	\> {\bf build} \\
	\> 1.\> compute $Y_0, Y_1, ..., Y_h$ (which are defined in \eqref{eqn:alg1:Y_i}) \\
	\> 2.\> {\bf for} each $i \in [0, h]$ {\bf do} \\
	\> 3.\>\> {\bf for} each point $p \in P$ {\bf do} \\
	\> 4.\>\>\> $S \leftarrow \set{y \in Y_i \mid D(p, y) \le \phi \cdot 2^i}$ \\
	\> 5.\>\>\> create an edge $(p, y)$ for each $y \in S$
}

Line 1 can be implemented in $O(n \log (n\Delta))$ time using an algorithm due to Har-Peled and Mendel \cite[Theorem 3.2]{hm06}. Next, we will concentrate on Lines 2-5.

\vgap

At Line 2, prior to entering Line 3, we create a data structure $T$ on $Y_i$ that allows us to answer 2-ANN queries on $Y_i$. The structure should be fully dynamic, i.e., it can support both insertions and deletions. Denote by $t_\qry$ the worst-case time for $T$ to answer a 2-ANN query, and by $t_\upd$ the worst-case time for $T$ to perform an insertion or deletion. Immediately, it follows that $T$ can be built in $O(|Y_i| \cdot t_\upd) = O(n \cdot t_\upd)$ time.

\vgap

At Line 4, we retrieve $S$ using $T$ as follows. Initially, $S = \emptyset$ and $T$ stores exactly the points in $Y_i$. We then repeatedly (i) find a 2-ANN $y$ of the point $p$ from $T$, (ii) add $y$ to $S$ if $D(p, y) \le \phi \cdot 2^i$, and (iii) delete $y$ from $T$. The repetition continues until $D(p, y) > 2\phi \cdot 2^i$ for the first time.

\vgap

We argue that the set $S$ thus computed is precisely the one needed at Line 4. Let $S_\del$ be the set of points removed from $T$, and $y_\last$ be the last point removed, i.e., $D(p, y_\last) > 2\phi \cdot 2^i$. If $S_\del$ misses a point $y' \in Y_i$ with $D(p, y) \le \phi \cdot 2^i$, then $y'$ must still remain in $T$. This, however, would contradict the fact that $y_\last$ is a 2-ANN of $p$ (among the points remaining in $T$) because $2 \cdot D(p, y') \le 2 \phi \cdot 2^i < D(p, y_\last)$.

\vgap

The retrieval of $S_\del$ incurs a running time of $O(|S_\del| \cdot (t_\qry + t_\upd))$. We argue that $|S_\del| = O(\phi^\lambda)$. Note that every point in $S_\del$ --- except $y_\last$ --- falls in $B(p, 2\phi \cdot 2^i)$. Thus, the diameter of $S_\del \setm \set{y_\last}$ is at most $4\phi \cdot 2^i$. On the other hand, because all the points of $S_\del$ come from $Y_i$, their inter-point distance is at least $2^i$. Hence, the aspect ratio of $S_\del \setm \set{y_\last}$ is at most $4\phi$. It then follows from Fact~\ref{fact:alg1:pg-proof:aspect} that $S_\del \setm \set{y_\last}$ has $O((4\phi)^\lambda) = O(\phi^\lambda)$ points.

\vgap

Prior to entering Line 5, we restore $T$ by inserting all the points of $S_\del$ back in $T$. This costs another $O(|S_\del| \cdot t_\upd)$ time. We can now conclude that Line 4 takes $O(\phi^\lambda \cdot (t_\qry + t_\upd))$ time for each point in $P$. As a result, the total running time of Lines 2-5 is $O(\phi^\lambda \cdot (t_\qry + t_\upd) \cdot n \cdot h)$.

\vgap

We have shown that {\bf build} runs in
\myeqn{
	O(n \log (n\Delta) + (1/\eps)^\lambda \cdot (t_\qry + t_\upd) \cdot n \log \Delta) \label{eqn:alg1:construction}
}
time overall. It remains to choose a good data structure for $T$. The structure of Cole and Gottlieb \cite{kl04} ensures $t_\qry = O(\log n)$ and $t_\upd = O(\log n)$. Plugging these bounds into \eqref{eqn:alg1:construction} gives the construction time claimed in Theorem~\ref{thm:ub-metric}.


\extraspacing {\bf Remark.} The above discussion has assumed that we know the minimum and maximum inter-point distances in $P$, denoted as $d_{min}$ and $d_{max}$, respectively (note: $d_{max} = \diam(P)$). The assumption can be removed using standard techniques \cite{hm06, kl04}. More specifically, we can obtain in $O(n \log n)$ time values $\hat{d}_{min} \in [\fr{1}{2} d_{min}, d_{min}]$ and $\hat{d}_{max} \in [d_{max}, 2d_{max}]$.\footnote{To compute $\hat{d}_{max}$, take an arbitrary point $p \in P$ and then set $\hat{d}_{max}$ $=$ $2 \max_{p' \in P} D(p, p')$. To compute $\hat{d}_{min}$, first build a 2-ANN structure on $P$. For each point $p \in P$, use the structure to find a 2-ANN $p'$ of $p$ and record the distance $D(p,p')$ for $p$. Then, $\hat{d}_{min}$ can be set to half of the smallest recorded distance of all points.} The ratio $\hat{d}_{max}/\hat{d}_{min}$ approximates the aspect ratio $\Delta$ up to a factor of 4. Our algorithm can then be applied after replacing $d_{min}$ and $\diam(P)$ with $\hat{d}_{min}$ and $\hat{d}_{max}$, respectively.

\section{A Size Lower Bound under $\bm{\eps = 1}$} \label{sec:lb1}

This section serves as a proof of Statement (1) of Theorem~\ref{thm:lbs}. Recall that we are given integers $\Delta$ and $n$ both of which are powers of 2; they satisfy the condition that $n \ge 2$ and $n^2 \le 2\Delta \le 2^n$.

\vgap

We will design a metric space by resorting to a complete binary tree $\T$ of $2\Delta$ leaves. The tree has $h + 1$ levels where $h = \log (2\Delta)$. We number the levels bottom-up, with the leaves at level 0 and the root at level $h$. For each node $u$ of $\T$, we use $\level(u)$ to represent its level. To each edge $\set{u, v}$ of $\T$ --- w.l.o.g., assume that $u$ is the parent of $v$ --- we assign a {\em weight} that equals 1 if $v$ is a leaf, or $2^{\level(v)-1}$ otherwise.

\vgap

We are now able to clarify the metric space $(\M, D)$:
\myitems{
	\item $\M$ is the set of leaves in $\T$;
	\item for any leaves $v_1, v_2$ in $\T$, their distance $D(v_1, v_2)$ equals the total weight of the edges on the unique simple path connecting $v_1$ and $v_2$ in $\T$.
}
When $v_1 \ne v_2$, our design of weights allows a simple calculation of $D(v_1, v_2)$: if the lowest common ancestor (LCA) of $v_1$ and $v_2$ is at level $\ell$, then $D(v_1, v_2) = 2^\ell$. It is easy to verify that $D$ satisfies identify of indiscernibles, symmetry, and triangle inequality. The doubling dimension of $(\M, D)$ is 1, as proved in Appendix~\ref{app:lb1:ddim}.

\begin{figure}
	\centering
    \includegraphics[height=55mm]{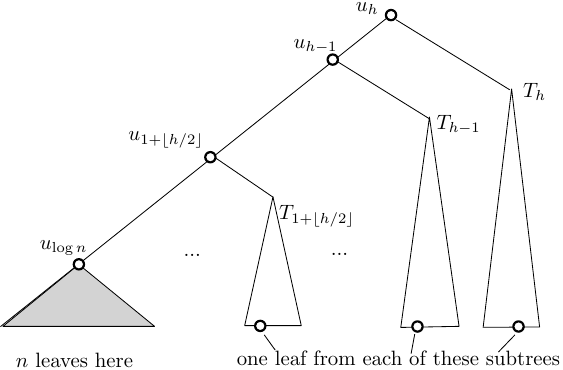}
    \figcapup
    \caption{Hard input for Section~\ref{sec:lb1}}
    \label{fig:lb1}
    \figcapdown
\end{figure}

\vgap

Next, we describe a set $P$ of $n$ points (a.k.a.\ leaves) from $\M$, which will serve as our hard input. Let $\pi$ be the leftmost root-to-leaf path of $\T$. For each $i \in [0, h]$, denote by $u_i$ be the level-$i$ node on $\pi$, and by $T_i$ the right subtree of $u_i$. We create $P$ as follows:
\myitems{
	\item For each $i \in (h/2, h]$, we add to $P$ one arbitrary leaf in $T_i$.
	\item Add to $P$ all the leaves in the subtree of $u_{\log n}$.
}
See Figure~\ref{fig:lb1} for an illustration. We will use $P_1$ (resp., $P_2$) to represent the set of leaves added in the first (resp., second) bullet. As $\log n \le \fr{1}{2} \log (2\Delta) = h/2$, the sets $P_1$ and $P_2$ are disjoint. The total size of $P$ is $n + \floor{h/2}$, which is between $n$ and $3n/2$ (recall that $2\Delta \le 2^n$ and hence $h \le n$). Note that $|P_1| = n$ and $|P_2| = \floor{h / 2} \ge 1$ because $2\Delta \ge n^2 \ge 4$. The reader can verify that $\diam(P) = 2^h = 2\Delta$ and the smallest inter-point distance of $P$ is 2; hence, the aspect ratio of $P$ is $\Delta$.

\vgap

Consider any 2-PG $G$ of $P$; by Fact~\ref{fact:alg1:pg-proof:navigable}, the graph $G$ needs to be 2-navigable. We will argue that $G$ must have an edge $(v_1, v_2)$ for every $(v_1, v_2) \in P_1 \times P_2$. As a result, the number of edges in $G$ must be at least $|P_1||P_2| = \Omega(n \log \Delta)$, as claimed in Statement (a) of Theorem~\ref{thm:lbs}.

\vgap

Assume, for contradiction, that $G$ has no edge $(v_1, v_2)$ for some $v_1 \in P_1$ and $v_2 \in P_2$. We will show that $G$ cannot be 2-navigable. W.l.o.g., assume that $v_2$ is in $T_\ell$ for some $\ell \in (h/2, h]$. It thus holds that $D(v_1, v_2) = 2^\ell$ because the LCA of $v_1$ and $v_2$ is $u_\ell$. Let us set $q = v_2$; as $q \in P$, the NN of $q$ is $v_2$ itself (with the NN-distance $D(q, v_2) = 0$). Hence, $v_1$ is not a 2-ANN of $q$. We claim that $v_1$ has no out-neighbor in $G$ that is closer to $q$ than $v_1$, because of which $G$ is not 2-navigable.

\vgap

Let $p_\out$ be an arbitrary out-neighbor of $v_1$; to prove the claim, it suffices to explain why $D(p_\out, q) \ge D(v_1, q)$. Clearly, $p_\out \neq v_2$ because the edge $(v_1, v_2)$ is absent in $G$. Where else can $p_\out$ be? If $p_\out$ is a descendant of $u_i$ for some $i \le \ell - 1$, then the LCA of $p_\out$ and $q = v_2$ must be $u_\ell$, because of which $D(p_\out, q) = 2^\ell = D(v_1, q)$. If $p_\out$ is in $T_k$ for some $i \ge \ell + 1$, then the LCA of $p_\out$ and $q = v_2$ is $u_i$, because of which $D(p_\out, q) = 2^i > D(v_1, q)$. As no other cases are possible, we conclude that $G$ is not 2-navigable.

\section{A Size Lower Bound under Small $\bm{\eps}$} \label{sec:lb2}

This section serves as a proof of Statement (2) of Theorem~\ref{thm:lbs}. Let us start by introducing several useful notations. Recall that the statement assumes that three integers $s \ge 2, t \ge 1$, and $d \ge 1$ have been given. We use the notation $\Zdom_s$ to represent the set $\set{0, 1, ..., s-1}$. Given two points $p, w \in \real^d$, we define the output of $p + w$ to be the point whose $i$-th coordinate is $p[i] + w[i]$ for each $i \in [d]$. If $S$ is a (possibly infinite) set of points in $\real^d$, given a point $w$, we define the {\em $w$-translated copy} of $S$ to be $\set{p + w \mid p \in S}$.

\vgap

Define $M = (\Zdom_s)^d$,
which is a set of $s^d$ points. Furthermore, define
\myeqn{
	W &=& \set{(i \cdot 2s, 0, 0, ..., 0) \mid i \in [0, t-1]}
	\label{eqn:lb2:W}
}
namely, every point $w \in W$ has a non-zero coordinate only on the first dimension, with $w[1]$ being a multiple of $2s$ in $[0, 2s (t-1)]$. For each $w \in W$, define
\myeqn{
	M_w &=& \text{the $w$-translated copy of $M$}. \nn
}
We will refer to each $M_w$ as a {\em block}.

\vgap

The hard data input we use is
\myeqn{
	P &=& \bigcup_{w \in W} M_w. \label{eqn:lb2:P}
}
See Figure~\ref{fig:lb2} for an illustration in $\real^2$. It is clear that $n = |P| = s^d \cdot t$.

\begin{figure}
	\centering
    \includegraphics[height=25mm]{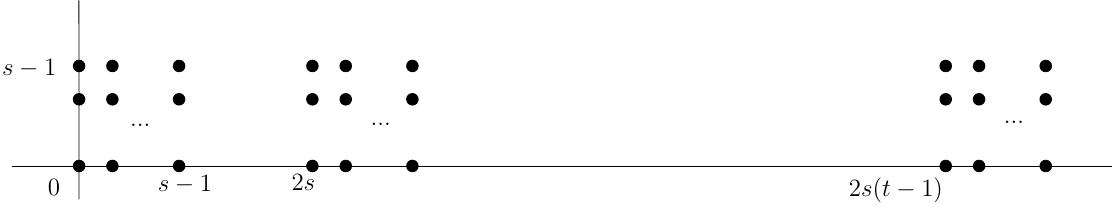}
    \figcapup
    \caption{Hard input for Section~\ref{sec:lb2}}
    \label{fig:lb2}
    \figcapdown
\end{figure}

\vgap

Next, we will design the metric space $(\M, D)$. The set $\M$ includes $P$ and one extra point $q$, which is non-Euclidean (i.e., $q$ is not in $\real^d$). The definition of $D$ --- which will be clarified later --- depends on $p^*$, which is a point in $P$. As we will see, varying the choice of $p^*$ will result in a different $D$. For this reason, we will represent the distance function as $D_{p^*}$. This gives rise to a set of distance functions:
\myeqn{
	\D &=& \set{D_{p^*} \mid p^* \in P}. \label{eqn:lb2:class-D}
}
We now specify the details of $D_{p^*}$. Denote by $w^*$ the unique point in $W$ such that $p^*$ is in the block $M_{w^*}$. Then:
\myitems{
	\item for any $p_1, p_2 \in P$, define $D_{p^*}(p_1, p_2) = D_{p^*}(p_2, p_1) = L_\infty(p_1, p_2)$;
	\item for any $p \in P \setm M_{w^*}$, define $D_{p^*}(p, q) = D_{p^*}(q, p) = L_\infty(p, w^*)$;
	\item for any $p \in M_{w^*}$ and $p \ne p^*$, define $D_{p^*}(p, q) = D_{p^*}(q, p) = s$;
	\item define $D_{p^*}(p^*, q) = D_{p^*}(q, p^*) = s-1$;
	\item define $D(q, q) = 0$.
}
We prove in Appendix~\ref{app:lmm:lb2:D}:

\begin{lemma} \label{lmm:lb2:D}
    For every $p^* \in P$, $(\M, D_{p^*})$ is a metric space with doubling dimension $\lambda \le \log (1 + 2^d)$.
\end{lemma}

Set $\eps = 1/(2s)$, as in Statement (2) of Theorem~\ref{thm:lbs}. When an algorithm constructs a $(1+\eps)$-PG $G$ of $P$, it has access only to the points in $P$, but not the non-Euclidean point $q$. Thus, the algorithm can evaluate (as it wishes) only the distances between the points in $P$, but not the distance between $q$ and any $p \in P$.

\vgap

The above observation gives rise to an adversarial argument. Imagine an adversary --- named Alice --- who does not finalize the function $D$ until after seeing the PG $G$ produced by the algorithm. Of course, Alice cannot lie: her ultimate choice of $D$ must be consistent with the distances already exposed to the algorithm. However, this will not be a problem as long as she chooses $D$ from the class $\D$ (see \eqref{eqn:lb2:class-D}), noticing that every $D \in \D$ gives exactly the same $D(p_1, p_2)$ for any $p_1, p_2 \in P$.

\vgap

Next, we will argue that, for every $w \in W$ (see \eqref{eqn:lb2:W} for $W$) and any distinct $p_1, p_2$ in the same block $M_w$, there must be an edge $(p_1, p_2)$ in $G$. As each block has $s^d$ points and there are $t$ blocks, the total number of edges in $G$ will then be at least
\myeqn{
	s^d \cdot (s^d - 1) \cdot t = \Omega(s^d \cdot n) \nn
}
because $s \ge 2$. This will establish Statement (2) of Theorem~\ref{thm:lbs}.

\vgap

Assume, for contradiction, that $G$ has no edge $(p_1, p_2)$ for some points $p_1$ and $p_2$ that appear in the same block $M_w$, for some $w \in W$. Seeing this, adversary Alice sets $p^*$ to $p_2$ and thereby finalizes $D$ to $D_{p_2}$. We will see that under the metric space $(\M, D_{p_2})$, the graph $G$ cannot be $(1+\eps)$-navigable, which by Fact~\ref{fact:alg1:pg-proof:navigable} means that $G$ cannot be a $(1+\eps)$-PG of $P$.

\vgap

Under $D_{p_2}$, the NN of $q$ is $p_2$ with $D(q, p_2) = s - 1$; indeed, by our design $D(q, p) \ge s$ for all the other points $p \in P$. Furthermore, as $p_1$ is from the same block as $p_2$, we have $D(q, p_1) = s$. This means that $p_1$ is not a $(1+\eps)$-ANN of $q$ because $s > s - \fr{1}{2} - \fr{1}{2s} = (s-1) (1 + \fr{1}{2s}) = (s-1) (1 + \eps)$. We claim that $p_1$ has no out-neighbor in $G$ that is closer to $q$ than $p_1$, because of which $G$ is not $(1+\eps)$-navigable.

\vgap

Let $p_\out$ be an arbitrary out-neighbor of $p_1$; to prove our  claim above, it suffices to explain why $D(p_\out, q) \ge D(p_1, q)$. Clearly, $p_\out \ne p_2$ because the edge $(p_1, p_2)$ is absent in $G$. If $p_\out$ is in $M_w$ (i.e., the block of $p_2$), then $D(p_\out, q) = s = D(p_1, q)$. If $p_\out$ is a block different from $M_w$, then $D(p_\out, q) = L_\infty(p_\out, w)$, which is at least $s+1$ and hence greater than $D(p_1, q)$. We thus conclude that $G$ is not $(1+\eps)$-navigable.

\vgap

Finally, let us note that, under any $D \in \D$, the maximum distance (under $D$) between two points in $S$ is less than $2st = O(n)$, while the smallest inter-point distance is 1. Hence, the aspect ratio of $P$ is $O(n)$. This completes the proof of Theorem~\ref{thm:lbs}.

\extraspacing {\bf Remark.} When $t = 1$, the set $P$ degenerates into a hard input used in \cite{hm06} to prove a lower bound on the {\em query time} of ANN data structures. Generalizing that hard input to establish a {\em size} lower bound for $(1+\eps)$-PGs under a wide range of $\eps$ demands additional ideas, as we have shown above.

\section{Smaller Proximity Graphs in the Euclidean Space} \label{sec:euc}

This section serves as a proof of Theorem~\ref{thm:ub-geo}. The goal is to improve the size bound of Theorem~\ref{thm:ub-metric} in the special metric space of $(\real^d, L_2)$ by shaving-off the $\log \Delta$ factor. Our discussion will assume that the smallest inter-point distance in $P$ is 2 (as can be achieved by scaling the dimensions of $\real^d$) and that the value of $\diam(P)$ is available. The assumption can be removed using the same techniques explained in the remark of Section~\ref{sec:alg1:construction}.

\vgap

Let us first apply the algorithm of Theorem~\ref{thm:ub-metric} to obtain a $(1+\eps)$-PG for the data input $P$ in $(\real^d, L_2)$; we will denote the graph as $G_\net$, where the subscript reminds us that it is obtained using an algorithm designed for general metric spaces. As analyzed in Section~\ref{sec:alg1:space-qry}, each vertex of $G_\net$ has an out-degree of $O((1/\eps)^\lambda \cdot \log \Delta)$ such that the graph has $O((1/\eps)^\lambda \cdot n \log \Delta)$ edges in total (recall that $\lambda$ is the doubling dimension of $(\real^d, L_2)$).

\vgap

Consider the following drastic idea to ``force'' the edge number to drop by a $\log \Delta$ factor:
\myitems{
	\item sample each vertex independently with probability
	\myeqn{
		\tau = \fr{z}{\log \Delta}
		\label{eqn:euc:tau}
	}
	where $z$ is a constant to be determined later;

	\item keep the edges of only the sampled vertices and discard all other edges (non-sampled vertices are retained, even though their out-degrees are now 0).
}
In expectation, the resulting graph --- denoted as $G'_\net$ --- has $O((1/\eps)^\lambda \cdot n)$ edges, exactly what we hope for. However, the idea does not work (yet) because $G'_\net$ may no longer be a $(1+\eps)$-PG of $P$.

\vgap

A second idea now kicks in: how about ``patching up'' $G'_\net$ by merging it with a ``small-but-slow'' $(1+\eps)$-PG $G_\geo$ that has $O((1/\eps)^\lambda \cdot n)$ edges but possibly very poor query time? The subscript of $G_\geo$ serves as a reminder that $G_\geo$ is indeed a blessing of geometry --- Statement (1) of Theorem~\ref{thm:lbs} has ruled out the existence of such a graph in every general metric space, no matter how bad the query time is.
Formally, the merging of $G'_\net$ and $G_\geo$ gives us a graph $G$ defined as follows:
\myitems{
	\item The vertices of $G$ have one-one correspondence to $P$ (recall that both $G'_\net$ and $G_\geo$ have the same vertex set, i.e., $P$).
	\item The out-edge set of each point $p \in P$ in $G$ is the union of those in $G'_\net$ and $G_\geo$.
}

\vslit

The rest of the section will develop the above ideas into a concrete algorithm to build a proximity graph meeting the requirements of Theorem~\ref{thm:ub-geo}.

\subsection{Small-but-Slow Proximity Graphs} \label{sec:euc:small-slow}

This subsection deals with the following problem: given a set $P$ of $n$ points in $\real^d$, build a $(1+\eps)$-PG $G_\geo$ of $P$ of $O((1/\eps)^\lambda \cdot n)$ edges under $L_2$ norm. Arya and Mount \cite{am93} have proven such a graph's existence, but their construction takes $\Omega((1/\eps)^d \cdot n^2)$ expected time. Our objective is to achieve a near-linear dependence on $n$ in construction time. To achieve the objective,  we will introduce a variant of the so-called ``$\theta$-graph'' known to permit fast construction. Then, we will prove that an appropriate choice of $\theta$ will guarantee that the graph is a $(1+\eps)$-PG of $P$.

\vgap

A {\em halfspace} in $\real^d$ is the set of points $\set{x \in \real^d \mid \sum_{i=1}^d x[i] \cdot c_i \ge c_{d+1}}$, where $c_1, c_2, ..., c_{d+1}$ are constants. The boundary of the halfspace is the plane $\sum_{i=1}^d x[i] \cdot c_i = c_{d+1}$. A set of halfspaces is said to be in {\em general position} if no two halfspaces have parallel boundary planes. A (simplicial) {\em cone} $C$ is the intersection of $d$ halfspaces in general position. The {\em apex} of $C$ is the intersection of the boundary planes of those halfspaces, and the {\em angular diameter} is the largest angle between two rays inside $C$ emanating from the apex; see Figure~\ref{fig:euc:cone}(a) for a 2D example.

\vgap

For any angle (measured in Radians) $\theta$ satisfying $0 < \theta < \pi$, Yao \cite{y82} gave an algorithm to compute in $O((1/\theta)^{d-1})$ time a set $\C$ of cones with the properties below:
\myitems{
	\item each cone in $\C$ has its apex at the origin and has an angular diameter at most $\theta$;
	\item the union of all cones in $\C$ is $\real^d$.
}
Let us associate each cone $C \in \C$ with an arbitrary ray --- denoted as $\rho_C$ --- that emanates from the origin and is contained in $C$.

\begin{figure}
	\centering
	\begin{tabular}{cc}
	 	\includegraphics[height=22mm]{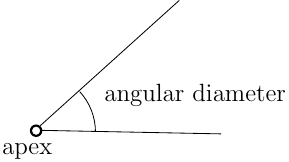} &
		\includegraphics[height=30mm]{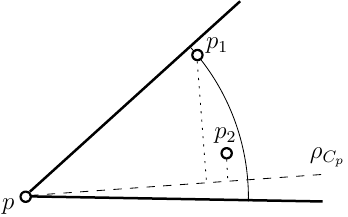} \\
		(a) Cone &
		(b) Nearest Point on Ray
	\end{tabular}
	\figcapup
	\caption{Key concepts underlying the $\theta$-graph}
	\label{fig:euc:cone}
	\figcapdown
\end{figure}

\vgap

Fix an arbitrary point $w \in \real^d$. For each $C \in \C$, we use $C_w$ to denote the $w$-translated copy of $C$ (the reader may wish to review Section~\ref{sec:lb2} for what is a ``$w$-translated copy'').
Define
\myeqn{
	\C_w &=& \set{C_w \mid C \in \C}. \label{eqn:euc:set_C_w}
}
The set $\C_w$ comprises $|\C| = O((1/\theta)^{d-1})$ cones with apex $w$ and angular diameter at most $\theta$ whose union covers $\real^d$. For each cone $C_w \in \C_w$, denote by $\rho_{C_w}$ the $w$-translated copy of the ray $\rho_C$. We will refer to $\rho_{C_w}$ the {\em designated ray} of cone $C_w$.

\vgap

Now, set $w$ to a point $p \in P$. We say that a cone $C_p \in \C_p$ is {\em non-empty} if $C_p$ covers at least one other point of $P$ besides $p$. For each non-empty $C_p$, identify a point $p'$ as the {\em nearest-point-on-ray} of $p$ in $C_p$ as follows:
\myitems{
	\item Let $S$ be the set of points in $P$ covered by $C_p$, after excluding $p$ itself.
	\item Project all the points of $S$ onto $\rho_{C_p}$, i.e., the designated ray of $C_p$.
	\item Then, $p'$ is the point whose projection on $\rho_{C_p}$ has the smallest $L_2$ distance to $p$.
}
Figure~\ref{fig:euc:cone}(b) illustrates an example where $S = \set{p_1, p_2}$. Point $p_1$ is the nearest-point-on-ray of $p$ in $C_p$ because its projection on ray $\rho_{C_p}$ is closer to $p$ than that of $p_2$ (note: $p_1$ actually has a greater $L_2$ distance from $p$ than $p_2$).

\vgap

We are ready to define the $\theta$-graph of $P$. This is a simple directed graph where
\myitems{
	\item the vertices have one-one correspondence to $P$;
	\item for any distinct points $p, p' \in P$, there is an edge from $p$ to $p'$ if and only if $p'$ is the nearest-point-on-ray of $p$ in some non-empty cone of $\C_p$.
}
Every vertex $p \in P$ in the $\theta$-graph has an out-degree at most $|\C| = O((1/\theta)^{d-1})$ (at most one out-edge from $p$ for each non-empty cone of $\C_p$). The total number of edges is thus $O((1/\theta)^{d-1} \cdot n)$. Such a graph can be constructed in $(1/\theta)^{d-1} \cdot n \polylog n$ time \cite{ams99b, rs91}.

\vgap

We prove the next lemma in Appendix~\ref{app:lmm:euc:theta-graph}.

\begin{lemma} \label{lmm:euc:theta-graph}
    A $(\eps/32)$-graph of $P$ is a $(1+\eps)$-proximity graph of $P$.
\end{lemma}

In the subsequent discussion, we will set $G_\geo$ to an $(\eps/32)$-graph of $P$, which can be constructed in $(1/\eps)^{d-1} \cdot n \polylog n$ time. Recall that the value $d$ never exceeds the doubling dimension $\lambda$.

\subsection{The Power of Merging} \label{sec:euc:merged}

We have obtained two graphs: $G'_\net$ and $G_\geo$. In particular, $G'_\net$ was obtained from $G_\net$ --- which itself was built using Theorem~\ref{thm:ub-metric} --- via vertex sampling. We call a point $p \in P$ a {\em jackpot} point/vertex if it was sampled (in the process of building $G'_\net$); recall that all the out-edges of $p$ in $G_\net$ are retained by $G'_\net$.

\vgap

Merging $G'_\net$ and $G_\geo$ gives graph $G$, which has $O((1/\theta)^\lambda \cdot n)$ edges in expectation as explained earlier. $G$ must be $(1+\eps)$-navigable (and hence a $(1+\eps)$-PG of $P$ by Fact~\ref{fact:alg1:pg-proof:navigable}). To see why, take an arbitrary point $p \in P$ and an arbitrary query $q \in \real^d$ such that $p$ is not an $(1+\eps)$-ANN of $q$. As $G_\geo$ is $(1+\eps)$-navigable (because it is a $(1+\eps)$-PG; see Lemma~\ref{lmm:euc:theta-graph}), there must exist an out-neighbor $p_\out$ of $p$ in $G_\geo$ with $L_2(p_\out, q) < L_2(p, q)$. Point $p_\out$ remains as an out-neighbor of $p$ in $G$, thus confirming that $G$ is $(1+\eps)$-navigable.

\vgap


Next, we will prove that w.h.p.\ the merged graph $G$ achieves a small query time for one single query point $q \in \real^d$. The next subsection will extend the result to all query points.

\vgap

Let us temporarily ignore $G'_\net$ and focus on $G_\geo$. For each $p \in P$, if we run {\bf greedy} on $G_\geo$ with parameters $p_\start = p$ and $q$, the algorithm visits a sequence of hop vertices (i.e., the $p^\circ$ vertices in the pseudocode in Section~\ref{sec:intro:prob}); let us denote that sequence as $\sigma_\geo(p)$. We say that $\sigma_\geo(p)$ is {\em long} if it has at least $\ln n \cdot \log \Delta$ vertices. The following is a condition we would like to have:

\vslit

\minipg{0.9\linewidth}{
	{\bf The jackpot condition:} Every long $\sigma_\geo(p)$ (where $p \in P$) encounters a jackpot point within the first $\ceil{\ln n \cdot \log \Delta}$ vertices.
}

\vslit

\noindent The jackpot condition holds w.h.p.. To see why, notice that, for each long $\sigma_\geo(p)$, the probability for none of the first $l$ vertices on $\sigma_\geo(p)$ to be sampled is at most $(1-\tau)^l \le e^{-\tau \cdot l}$, which is at most $1/n^z$ for $l = \ceil{\ln n \cdot \log \Delta}$ and the value of $\tau$ in \eqref{eqn:euc:tau}. As there are at most $n$ long sequences (at most one for each $p \in P$), the probability that all of them obey the jackpot condition is at least $1 - 1/n^{z-1}$.

\vgap

The remainder of this subsection will prove that, under the jackpot condition, the merged graph $G$ guarantees a query time of $O((1/\eps)^\lambda \cdot \log^2 \Delta + (1/\eps)^{d-1} \log n \cdot \log^2 \Delta)$ for $q$. Run {\bf greedy} on $G$ with an arbitrary $p_\start$ and stop the algorithm after it has visited
\myeqn{
	k = 1 + \ceil{\log (2\Delta)} \nn
}
hop vertices that are jackpot points (provided that it has not already self-terminated). Denote by $\sigma$ the sequence of hop vertices visited by {\bf greedy}. Chop $\sigma$ into subsequences, each of which (i) either ends at a jackpot vertex or is the last subsequence of $\sigma$, and (ii) includes no jackpot vertex except possibly at the end.

\begin{lemma} \label{lmm:euc:subseq-len}
	Every subsequence has at most $\ceil{\ln n \cdot \log \Delta}$ vertices.
\end{lemma}

\begin{proof}
	Consider any subsequence $\sigma'$ of length at least 2. Let $p_\first$ be the first vertex of $\sigma'$, which must be a non-jackpot point. Observe that $\sigma'$ must be a prefix of $\sigma_\geo(p_\first)$. To see why, take any vertex $p$ on $\sigma'$ except the last vertex of $\sigma'$. As $p$ is not a jackpot point, all its out-edges originate from $G_\geo$. When $p$ is the hop vertex, {\bf greedy} must choose the same next hop as it would when running on $G_\geo$, explaining why $\sigma'$ is a prefix of $\sigma_\geo(p_\first)$.

	\vgap

	Because $\sigma'$ has at most one jackpot vertex, its length must be at most $\ceil{\ln n \cdot \log \Delta}$ under the jackpot condition.
\end{proof}

If {\bf greedy} terminates without seeing $k$ jackpot hop vertices, it must return a $(1+\eps)$-ANN of $q$ because $G$ is a $(1+\eps)$-PG. Next, we consider the situation where {\bf greedy} is forced to terminate. We will argue that $\sigma$ must contain at least one $(1+\eps)$-ANN of $q$. This implies that the last vertex of $\sigma$ must be a $(1+\eps)$-ANN of $q$ because the vertices on $\sigma$ have descending distances to $q$.

\vgap

Assume, for contradiction, that no vertex on $\sigma$ is a $(1+\eps)$-ANN of $q$. Denote by $p^*$ the NN of $q$.
 Since we manually stopped {\bf greedy}, the sequence $\sigma$ consists of exactly $k$ subsequences, each of which ends with a jackpot point. For each $i \in [k]$, define
\myeqn{
	p^\circ_i &=& \text{the last vertex of the $i$-th subsequence} \nn
}

\begin{lemma} \label{lmm:euc:help}
    For each $1 \le i \le \ceil{\log (2\Delta)}$, it holds that $\ceil{\log L_2(p^\circ_i, p^*)} > \ceil{\log L_2(p^\circ_{i+1}, p^*)}$.
\end{lemma}

\begin{proof}
    Let us write out the vertices of the $(i+1)$-th subsequence of $\sigma$ as $v_1, v_2, ..., v_l$ for some $l \le \ceil{\ln n \cdot \log \Delta}$. Note that $v_1$ succeeds $p^\circ_i$ in $\sigma$ and $v_l = p^\circ_{i+1}$. As $G$ is $(1+\eps)$-navigable and no vertex on $\sigma$ is a $(1+\eps)$-ANN of $q$, the $(1+\eps)$-navigable definition tells us $L_2(v_1, q) > L_2(v_2, q) > ... > L_2(v_l, q)$.

    \vgap

    Because $p^\circ_i$ is not a $(1+\eps)$-ANN of $q$, by Lemma~\ref{lmm:alg1:pg-proof:help2} its out-degree in $G_\net$ is at least 1. Let $p^+_\out$ be the vertex defined in \eqref{eqn:alg1:p+_out}, i.e., the out-neighbor of $p^\circ_i$ in $G_\net$ closest to $q$.     Because $p^\circ_i$ is a jackpot point, $p^+_\out$ must be an out-neighbor of $p^\circ_i$ in $G$. As {\bf greedy} always hops to the out-neighbor closest to $q$, we have $L_2(v_1, q) \le L_2(p^+_\out, q)$, which leads to $L_2(v_l, q) \le L_2(p^+_\out, q)$.

    \vgap

    Now, apply Statement (2) of Lemma~\ref{lmm:alg1:pg-proof:help2} by setting $p^\circ = p^\circ_i$, $\varrho = v_l$, and $D = L_2$. The application yields $\ceil{\log L_2(v_l, p^*)} < \ceil{\log L_2(p^\circ_i, p^*)}$, as claimed.
\end{proof}

As $L_2(p^\circ_1, p^*) \le \diam(P) = 2\Delta$ (recall that the smallest inter-point distance in $P$ is 2), Lemma~\ref{lmm:euc:help} implies that $L_2(p^\circ_k, p^*)$ must be strictly less than 2, indicating that $p^\circ_k = p^*$. This contradicts the fact that $\sigma$ contains no $(1+\eps)$-ANNs of $q$.

\vgap

Recall that each jackpot vertex has an out-degree of $O((1/\eps)^\lambda \cdot \log \Delta)$ in $G$, while each non-jackpot vertex has an out-degree of $O((1/\eps)^{d-1}))$ in $G$. Our algorithm  visits $O(\log \Delta)$ jackpot vertices and $O(\log n \cdot \log^2 \Delta)$ non-jackpot vertices (due to Lemma~\ref{lmm:euc:subseq-len}). The total query time is therefore
$O((1/\eps)^\lambda \cdot \log^2 \Delta + (1/\eps)^{d-1} \cdot \log n \cdot \log^2 \Delta)$.

\subsection{Achieving High Probability} \label{sec:euc:hp}

So far our query time holds w.h.p.\ on only one query point in $\real^d$. To prove Theorem~\ref{thm:ub-geo}, we must argue that w.h.p.\ the same query time holds on all query points in $\real^d$. The key observation that makes this possible is that, even though there are infinitely many query points,  only $O(n^{2d})$ representative ones need to be considered.

\vgap

The execution of {\bf greedy} is decided by the outcome of distance comparisons of the form ``{\em which of $L_2(p_1, q)$ and $L_2(p_2, q)$ is larger?}'' Imagine two query points $q_1$ and $q_2$ with the property:

\minipg{0.9\linewidth}{
	for any distinct points $p_1, p_2 \in P$: $L_2(p_1, q_1) < L_2(p_2, q_1) \Leftrightarrow L_2(p_1, q_2) < L_2(p_2, q_2)$.
}

\noindent For any $p_\start \in P$, the behavior of {\bf greedy} invoked with parameters $(p_\start, q_1)$ is exactly the same as invoked with $(p_\start, q_2)$. This is true regardless of which proximity graph is adopted.

\vgap

The points in $P$ define ${n \choose 2}$ perpendicular bisectors, which dissect $\real^d$ into $O(n^{2d})$ polytopes. Queries in each polytope have the same NN and induce the same behavior of {\bf greedy}. Take a query representative from each polytope. Section~\ref{sec:euc:merged} has shown that the merged graph $G$ guarantees a low query time on one query with probability at least $1 - 1/n^{z-1}$, where $z$ is the constant in \eqref{eqn:euc:tau}. We can thus conclude that $G$ guarantees a low query time on all the representatives (and hence all queries in $\real^d$) with probability at least $1 - n^{O(d)} /n^{z-1}$, which is greater than $1 - 1/n^c$ for any constant $c$ by making $z$ sufficiently large.

\vgap

Only one issue remains. Currently, the number of edges in $G$ is $O((1/\eps)^\lambda \cdot n)$ in expectation. To ensure this size bound w.h.p., we  run our construction algorithm $z' \cdot \log n$ times for a sufficiently large constant $z'$. With probability at least $1 - z' \log n \cdot n^{O(d)} /n^{z-1}$, the proximity graphs produced by all the runs guarantee query time $O((1/\eps)^\lambda \cdot \log^2 \Delta + (1/\eps)^{d-1} \log n \cdot \log^2 \Delta)$ on all queries in $\real^d$. By Markov's inequality, in each run, the graph size  exceeds twice the expectation with probability at most $1/2$. Therefore, the probability for the smallest $G$ of all runs to have $O((1/\eps)^\lambda \cdot n)$ edges is at least $1 - 1/n^{z'}$.

\vgap

We now conclude that w.h.p.\ we can compute in $(1/\eps)^\lambda \cdot n \polylog(n \Delta)$ time a $(1+\eps)$-PG that has $O((1/\eps)^\lambda \cdot n)$ edges and ensures query time $O((1/\eps)^\lambda \cdot \log^2 \Delta + (1/\eps)^{d-1} \log n \cdot \log^2 \Delta)$ for all queries. This completes the proof of Theorem~\ref{thm:ub-geo}.

\bibliographystyle{plain}
\bibliography{ref}

\appendix

\def\vgap{\vspace{2mm}}
\def\vslit{\vspace{1mm}}
\def\extraspacing{\vspace{4mm} \noindent}

\section{Proof of Fact~\ref{fact:alg1:pg-proof:navigable}}
\label{app:alg1:pg-proof:navigable}

\noindent{\bf The If-Direction ($\Leftarrow$).} To prove this direction, we assume that $G$ is $(1+\eps)$-navigable. Fix an arbitrary query point $q \in \M$ and an arbitrary data point $p_\start \in P$. Let $p$ be the point returned by the {\bf greedy} algorithm when invoked with parameters $(p_\start, q)$. Suppose that $p$ is not a $(1+\eps)$-ANN of $q$. Because $G$ is $(1+\eps)$-navigable, $p$ must have an out-neighbor closer to $q$ than $p$ itself, meaning that the greedy algorithm cannot terminate at $p$, giving a contradiction. Hence, $p$ must be a $(1+\eps)$-ANN of $q$; and thus $G$ is a $(1+\eps)$-PG.

\extraspacing{\bf The Only-If Direction ($\Rightarrow$).} To prove this direction, we assume that $G$ is a $(1+\eps)$-PG. Fix an arbitrary query point $q \in \M$ and an arbitrary data point $p\in P$ such that $p$ is not a $(1+\eps)$-ANN of $q$. Run the {\bf greedy} algorithm with $p_\start = p$ and $q$. Because $G$ is a $(1+\eps)$-PG, the algorithm must return a $(1+\eps)$-ANN of $q$ and hence cannot return $p_\start$. As a result, {\bf greedy} must be able to identify an out-neighbor $p_\out$ of $p_\start$ with $D(p_\out, q) < D(p_\start, q)$. The presence of $p_\out$ indicates that $G$ is $(1+\eps)$-navigable.

\section{Proof of Fact~\ref{fact:alg1:pg-proof:aspect}}
\label{app:alg1:pg-proof:aspect}
%
%

Let $d_\mit{min}$ and $d_\mit{max}$ be the minimum and maximum inter-point distances in $X$, respectively. Hence, $A = d_\mit{max}/d_\mit{min}$. By definition of diameter, the set $X$ can be covered by a ball $B(p, d_{max})$ where $p$ can be any point in $X$. Inductively, suppose that $B(p, d_{max})$ can be covered by $2^{i \cdot \lambda}$ balls of radius $d_\mit{max}/2^i$ for some $i \ge 0$. By definition of doubling dimension, we can cover each of those balls with $2^\lambda$ balls of radius $d_\mit{max}/2^{i+1}$. This means that $B(p, d_{max})$ can be covered by $2^{(i+1) \cdot \lambda}$ balls of radius $d_\mit{max}/2^{i+1}$.

\vgap

The above argument tells us that $B(p, d_{max})$ and, hence, $X$ can be covered by $2^{k \cdot \lambda}$ balls of radius $A \cdot d_\mit{min}/2^k$ for any $k \ge 0$. Now, set $k = 2 + \ceil{\log A}$, with which we have
\myeqn{
	\fr{d_{max}}{2^{k}}
	=
	\fr{A \cdot d_{min}}{2^{k}}
	<
	\fr{d_\mit{min}}{2}. \nn
}
Therefore, $X$ can be covered by
\myeqn{
	2^{k \lambda}
	\le 2^{\lambda \cdot (3 + \log A)}
	= (2^{3 + \log_2 A})^{\lambda} = (8A)^\lambda \nn
}
balls of radius less than $d_{min} / 2$. Each ball can cover at most one point in $X$. This is because the maximum distance of two points in a ball is less than $d_{min}$, which, let us recall, is the smallest inter-point distance in $X$. It thus follows that $|X| \le (8A)^\lambda$, which is $O(A^\lambda)$ for $\lambda = O(1)$.

\section{Doubling Dimension of the Hard Input in Section~\ref{sec:lb1}}
\label{app:lb1:ddim}

Given an arbitrary ball $B(p, r)$ where $p \in \M$ and $r > 0$, we will explain how to cover it with at most two balls of radius $r / 2$. This indicates that the metric space $(\M, D)$ has doubling dimension 1.

\vgap

Let us start by reminding the reader that, for two distinct points $v_1, v_2 \in \M$ (which are leaves of the binary tree $\T$), their distance is $2^\ell$, where $\ell \ge 1$ is the level of the LCA of $v_1$ and $v_2$. Therefore, if $r < 2$, then $B(p, r)$ contains only $p$ itself and thus can be covered by a single ball of radius $r/2$. The subsequent discussion assumes $r \ge 2$.

\vgap

Let $r' = 2^\ell$ be the largest power of 2 within the range $[2, 2\Delta]$ that does not exceed $r$. We have $B(p, r') = B(p, r)$ because every inter-point distance in $\M$ is a power of 2, as mentioned. Note that $\ell$ is some integer between 1 and $h = \log(2\Delta)$. Thus, point $p$, which is a leaf of $\T$, has an ancestor at level $\ell$, which we denote as $u_\anc$. The ball $B(p, r')$ is precisely the set of leaves in the subtree of $u_\anc$.

\vgap

Denote by $u_1$ and $u_2$ the left and right children of $u_\anc$, respectively. Let $X_1$ (resp., $X_2)$ be the set of leaves in the subtree of $u_1$ (resp., $u_2$). Clearly, $B(p, r') = X_1 \cup X_2$. Each of $X_1$ and $X_2$ is covered by a ball of radius $r' / 2 \le r / 2$. By symmetry, it suffices to prove this only for $X_1$. Take an arbitrary leaf $v$ from $X_1$. We argue that $X_1 \subseteq B(v, r'/2)$. Indeed, for any $v' \in X_1$ that differs from $v$, the LCA of $v$ and $v'$ must be a descendant of $u_1$ and hence must be at an level at most $\ell - 1$, meaning that $D(v, v') \le 2^{\ell - 1} = r'/2$.

\section{Proof of Lemma~\ref{lmm:lb2:D}} \label{app:lmm:lb2:D}

First, let us note the following property of our design: if two points $p_1, p_2 \in P$ are from different blocks, then $D_{p^*}(p_1, p_2) = L_\infty(p_1, p_2) = |p_1[1] - p_2[1]| \ge s + 1$.

\extraspacing {\bf Triangle Inequality.} To prove $(\M, D_{p^*})$ is a metric space, it suffices to prove that $D_{p^*}$ satisfies the triangle inequality.
Consider any $p_1, p_2$, and $p_3 \in \M$. If all of them originate from $P$, then their distances under $D_{p^*}$ are the same as under $L_\infty$-norm. Hence, we must have $D_{p^*}(p_1, p_2) \le D_{p^*}(p_1, p_3) + D_{p^*}(p_2, p_3)$. Next, we will assume that $p_3 = q$. Furthermore, we assume $p_1 \ne p_2$ because otherwise $D_{p^*}(p_1, p_2) = 0$ and the triangle inequality holds on $D_{p^*}(p_1, p_2)$, $D_{p^*}(p_1, q)$, and $D_{p^*}(p_2, q)$.

\vgap

If neither $p_1$ nor $p_2$ is from $M_{w^*}$ (i.e., the block of $p^*$), then the mutual distances of $p_1, p_2$, and $q$ under $D_{p^*}$ are the same as those of $p_1, p_2$, and $w^*$ under $L_\infty$. Those mutual distances must satisfy the triangle inequality.

\vgap

Consider now the case where both $p_1$ and $p_2$ are from $M_{w^*}$. W.l.o.g., suppose that $D_{p^*}(p_1, q) \le D_{p^*}(p_2, q)$. Thus, $D_{p^*}(p_1, p_2) \in [1, s - 1]$, $D_{p^*}(p_1, q) = s - 1$ or $s$, while $D_{p^*}(p_2, q) = s$. The three distances obey the triangle inequality.

\vgap

It remains to examine the case where $p_1 \in M_{w^*}$ but $p_2 \notin M_{w^*}$. We must have $D_{p^*}(p_1, q) = s - 1$ or $s$, $D_{p^*}(p_1, p_2)   > s$, and $D_{p^*}(p_2, q) = L_\infty(p_2, w^*) > s$. Let us first derive
\myeqn{
	D_{p^*}(p_1, q) + D_{p^*}(p_1, p_2)
	&\ge&
	s - 1 + L_\infty(p_1, p_2) \nn \\
	&=&
	s - 1 + |p_1[1] - p_2[1]| \nn \\
	\explain{as $p_1 \in M_{w^*}$, $w^* \in M_{w^*}$}
	&\ge&
	|p_1[1] - w^*[1]| + |p_1[1] - p_2[1]| \nn \\
	&\ge&
	|p_2[1] - w^*[1]| \nn \\
	&=&
	D_{p^*}(p_2, q). \nn
}
Similarly, we can derive
\myeqn{
	D_{p^*}(p_1, q) + D_{p^*}(p_2, q)
	&\ge&
	s - 1 + L_\infty(p_2, w^*) \nn \\
	&=&
	s - 1 + |p_2[1] - w^*[1]| \nn \\
	&\ge&
	|p_1[1] - w^*[1]| + |p_2[1] - w^*[1]| \nn \\
	&\ge&
	|p_1[1] - p_2[1]| \nn \\
	&=&
	D_{p^*}(p_1, p_2). \nn
}
We now conclude that $D_{p^*}(p_1, q)$, $D_{p^*}(p_1, p_2)$, and $D_{p^*}(p_2, q)$ satisfy the triangle inequality.

\extraspacing {\bf Doubling Dimension.} We will speak about balls under two different metric spaces: $(\M, D_{p^*})$ and $(P, L_\infty)$. To avoid confusion, we will adopt the notations below:
\myitems{
	\item given a point $p \in \M$, let $B_{p^*}(p, r)$ be the ball $B(p, r)$ under $(\M, D_{p^*})$;
	\item given a point $p \in P$, let $B_\infty(p, r)$ be the ball $B(p, r)$ under $(P, L_\infty)$.
}
When $p$ comes from $P$, we will refer to $B_{p^*}(p, r)$ as the {\em $D_{p^*}$-corresponding ball} of $B_\infty(p, r)$. These two balls have the following relationship:
\myitems{
	\item $B_\infty(p, r) \subseteq B_{p^*}(p, r)$;
	\item if $B_\infty(p, r) \ne B_{p^*}(p, r)$, then $B_{p^*}(p, r)$ contains only one extra point --- namely, $q$ --- outside of $B_\infty(p, r)$.
}

The metric space $(\real^d, L_\infty)$ is known to have doubling dimension  $d$. As $P \subseteq \real^d$, the metric space $(P, L_\infty)$ has doubling dimension at most $d$. We will utilize this fact to analyze the doubling dimension $\lambda$ of $(\M, D_{p^*})$. Given an arbitrary ball $B_{p^*}(p, r)$, we will show how to cover $B_{p^*}(p, r)$ with at most $1 + 2^d$ balls of radius $r/2$ under $(\M, D_{p^*})$, meaning that $\lambda \le \log (1 + 2^d)$.

\vgap

Consider first $p \in P$ (in other words, $p \ne q$). If $B_{p^*}(p, r) = B_\infty(p, r)$, we cover $B_{p^*}(p, r)$ with a set $S$ of at most $2^d$ balls under $(\M, D_{p^*})$ found using the procedure below:
\myitems{
	\item Initialize $S$ to be the empty set.

	\item Cover $B_\infty(p, r)$ with at most $2^d$ balls of radius $r/2$ under $(P, L_\infty)$.

	\item For each of the above ball, add its $D_{p^*}$-corresponding ball to $S$.
}
If $B_{p^*}(p, r) \ne B_\infty(p, r)$, we cover $B_{p^*}(p, r)$ with a set $S$ of at most $1 + 2^d$ balls under $(\M, D_{p^*})$ as follows:
\myitems{
	\item Obtain a set $S$ using the procedure for the case $B_{p^*}(p, r) = B_\infty(p, r)$.
	\item Add to $S$ the ball $B_{p^*}(q, r/2)$.
}

\vslit

The subsequent discussion will focus on the scenario where $p = q$:
\myitems{
	\item If $r < s-1$, then $B_{p^*}(q, r) = \set{q}$ and, hence, can be covered with a single ball of radius $r/2$ under $(\M, D_{p^*})$.


	\item If $r = s-1$, then $B_{p^*}(q, r) = \set{q, p^*}$ and, hence, can be covered with two balls of radius $r/2$ under $(\M, D_{p^*})$.



	\item If $r \ge s$, then $B_{p^*}(q,r) = B_{p^*}(w^*, r) = \set{q} \cup B_\infty(w^*, r)$; recall that $w^*$ is the point in $W$ such that $p^* \in M_{w^*}$. As $w^* \in P$, we have already explained how to cover $B_{p^*}(w^*, r)$ with at most $1 + 2^d$ balls under $(\M, D_{p^*})$. The same approach therefore works for  $B_{p^*}(q,r)$.
}

\vslit

We now complete the proof of Lemma~\ref{lmm:lb2:D}.

\section{Proof of Lemma~\ref{lmm:euc:theta-graph}} \label{app:lmm:euc:theta-graph}

\subsection{Basic Facts}

Let us start with three facts that will be useful in our technical derivation.

\begin{fact} \label{fact:app:lmm:euc:theta-graph:fact1}
	For any $0 \le x \le 1/2$, we have $\tan x \le 2x$.
\end{fact}

\begin{proof}
    Define $f(x) = \tan x - 2x$. Thus, $f'(x) = 1/(\cos x)^2 - 2$, which is negative for $0 \le x \le 1/2$. The fact then follows from $f(0) = 0$.
\end{proof}

Given two points $p, q$, we will use the following notation frequently:
\myeqn{
	\rho_{u, v}
	=
	\text{the ray that emanates from $p$ and passes $q$}.
	\label{eqn:app:lmm:euc:theta-graph:rho}
}


\vgap

\begin{fact} \label{fact:app:euc:lmm:theta-graph:fact2}
	Let $a, b,$ and $c$ be three distinct points in $\real^d$ such that the angle  $\gamma$ between rays $\rho_{a, b}$ and $\rho_{a, c}$ satisfies $0 < \gamma < \pi/2$. If $L_2(a, b) = L_2(a, c) = l > 0$, then $L_2(b, c)< l \cdot \tan \gamma$.
\end{fact}
\begin{proof}
	By applying basic geometric reasoning to the isosceles triangle $abc$, we obtain $L_2(b, c) = 2 \cdot l \sin(\gamma/2)$. Next, we will prove $2 \sin(\gamma/2) < \tan \gamma$.

	\vgap

	As $0 < \gamma/2 < \pi/4$, we have
	\myeqn{
		&& (2\cos(\gamma/2) + 1) (\cos(\gamma/2) - 1) < 0 \nn \\
		&\Rightarrow& 2\cos^2(\gamma/2) - \cos(\gamma/2) - 1 < 0 \nn \\
		&\Rightarrow& \cos^2 (\gamma/2) - \sin^2 (\gamma/2) < \cos (\gamma/ 2)  \nn 
	}
	As $\cos^2(\gamma/2) > \sin^2(\gamma/2)$ when $0 < \gamma/2 < \pi/4$, we can derive from the above
	\myeqn{
		&&
		1
		<
		\fr{\cos (\gamma/2)}{\cos^2 (\gamma/2) - \sin^2 (\gamma/2)} \nn \\
		&\Rightarrow&
		2 \cdot \sin(\gamma/2)
		<
		\frac{2 \sin \left( \gamma/2 \right) \cos \left( \gamma/2 \right)}{\cos^2 \left( \gamma/2 \right) - \sin^2 \left( \gamma/2 \right)}
		=
		\fr{\sin \gamma}{\cos \gamma} \nn
	}
	which is $\tan \gamma$.
\end{proof}


\begin{fact} \label{fact:app:euc:lmm:theta-graph:fact3}
    If $0 \le \gamma \le \eps/32$, then $(2+\eps) \cdot (2\tan\gamma + 1 - \cos \gamma) < \eps$.
\end{fact}

\begin{proof}
	Because $\gamma \le \eps/32 \le 1/32$, we have
	from Fact~\ref{fact:app:lmm:euc:theta-graph:fact1}:
	\myeqn{
		\tan \gamma \le 2\gamma \le \eps/16. \label{eqn:app:euc:lmm:theta-graph:fact3:0}
	}
	Later, we will prove:
	\myeqn{
		1-\cos \gamma < \eps/6.
		\label{eqn:app:euc:lmm:theta-graph:fact3:2}
	}
	Hence
	\myeqn{
		(2+\eps) \cdot (2\tan\gamma + 1 - \cos \gamma)
		&<&
		(2+\eps) \cdot \left( 2 \cdot \fr{\eps}{16} + \fr{\eps}{6} \right) \nn \\
		&=&
		(2+\eps) \cdot \fr{7\eps}{24} \nn \\
		\explain{as $0 < \eps \le 1$}
		&\le&
		3 \cdot \fr{7\eps}{24} \nn
	}
	which is less than $\eps$,
	as claimed.

	\vgap

	It remains to explain why \eqref{eqn:app:euc:lmm:theta-graph:fact3:2} is correct. First, using \eqref{eqn:app:euc:lmm:theta-graph:fact3:0}, we can derive:
	\myeqn{
		\cos^2 \gamma = \fr{1}{1+\tan^2 \gamma}
		\ge
		\fr{16^2}{16^2 + \eps^2}.
		\label{eqn:app:euc:lmm:theta-graph:fact3:1}
	}
	Define $f(x) = x^3 - 12 x^2 + 292 x - 3072$. We have $f'(x) = 3x^2-24x+292$, which is always positive. As $f(1) < 0$, we can assert that $f(x)  < 0$ for all $x \le 1$. This yields:
	\myeqn{
		\eps^3 - 12 \eps^2 + 292 \eps
		&<&
		3072
		\nn \\
		\Rightarrow
		\eps^4 - 12 \eps^3 + 292 \eps^2
		&<&
		3072\eps
		\nn
	}
	Rearranging terms from the above gives
	\myeqn{
		&& \fr{16^2}{16^2 + \eps^2} > (1-\eps/6)^2 \nn \\
		\explain{by \eqref{eqn:app:euc:lmm:theta-graph:fact3:1}}
		&\Rightarrow&
		\cos^2 \gamma > (1-\eps/6)^2 \nn \\
		\explain{as $\cos \gamma > 0$ and $\eps < 1$}
		&\Rightarrow&
		\cos \gamma > 1-\eps/6 \nn
	}
	thus giving the claim in \eqref{eqn:app:euc:lmm:theta-graph:fact3:2}.
\end{proof}

\subsection{The Proof}

Let $G$ be an $(\eps/32)$-graph of $P$. Our objective is to prove that $G$ is a $(1+\eps)$-PG of $P$. Fix an arbitrary data point $p \in P$ and an arbitrary query point $q \in \real^d$ such that $p$ is not a $(1+\eps)$-ANN of $q$. We will show that $p$ has an out-neighbor $p_\out$ in $G$ satisfying $L_2(p_\out, q) < L_2(p, q)$. This indicates that $G$ is $(1+\eps)$-navigable and thus a $(1+\eps)$-PG of $P$ by Fact~\ref{fact:alg1:pg-proof:navigable}.

\vgap

\begin{figure}
	\centering
	\includegraphics[height=60mm]{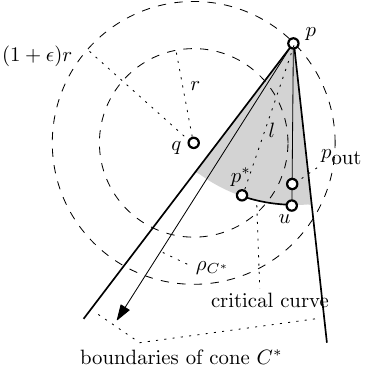}
	\caption{Case 1 of the proof in Section~\ref{app:lmm:euc:theta-graph}}
	\label{fig:app:euc:lmm:theta-graph:case1}
\end{figure}

Let us introduce two notions related to balls. Given a ball $B(p, r)$, we define its {\em surface} as the set $\set{x \in \real^d \mid L_2(p, x) = r}$. In addition, we say that a point $x \in \real^d$ is
\myitems{
	\item {\em in the interior} of $B(p, r)$ if $L_2(p, x)< r$;
	\item {\em on the surface} of $B(p, r)$ if $L_2(p, x) = r$.
}

\vslit

Let $p^*$ be the (exact) NN of $q$. Henceforth, we will fix
\myeqn{
	r = \fr{L_2(p, q)}{1+\eps}. \label{eqn:app:euc:lmm:theta-graph:r}
}
As $p$ is not a $(1+\eps)$-ANN of $q$, we must have $L_2(q, p^*) < r$, i.e., $p^*$ is in the interior of $B(q, r)$.

\vgap

Recall from Section~\ref{sec:euc:small-slow} that the union of the cones in $\C_p$ is $\real^d$ (see \eqref{eqn:euc:set_C_w} for the definition of $\C_p$). Hence, there must be a cone $C^* \in \C_p$ covering $p^*$. Define
\myeqn{
	p_\out &=&
	\text{the nearest-point-on-ray of $p$ in cone $C^*$}.
	\label{eqn:app:euc:lmm:theta-graph:p-out}
}
As explained in Section~\ref{sec:euc:small-slow}, $p_\out$ is the point whose projection on $\rho_{C^*}$ --- the designated ray of $C^*$ --- is the closest to $p$ under $L_2$ norm among the projections of all the points of $P \setm \set{p}$ covered by $C^*$; see Figure~\ref{fig:app:euc:lmm:theta-graph:case1} for an illustration. By the definition of $\theta$-graph, $p_\out$ is an out-neighbor of $p$ in $G$.

\vgap

The rest of the proof will show
\myeqn{
	L_2(p_\out, q) < (1+\eps) r
	\label{eqn:app:euc:lmm:theta-graph:goal}
}
which, by the value of $r$ in \eqref{eqn:app:euc:lmm:theta-graph:r}, says $L_2(p_\out, q) < L_2(p, q)$, thus leading us to the conclusion that $G$ is $(1+\eps)$-navigable.
Define
\myeqn{
	l &=& L_2(p, p^*)
	\label{eqn:app:euc:lmm:theta-graph:l}
}
Next, we proceed differently depending on the relationship between $L_2(p, p_\out)$ and $l$.

\extraspacing{\bf Case 1: $\bm{L_2(p, p_\out) \le l}$.} That is, $p_\out$ is in $B(p, l)$. This is the scenario illustrated in Figure~\ref{fig:app:euc:lmm:theta-graph:case1}. Before proceeding, the reader may wish to review the definition in \eqref{eqn:app:lmm:euc:theta-graph:rho} first.

\begin{figure}
	\centering
    \includegraphics[height=50mm]{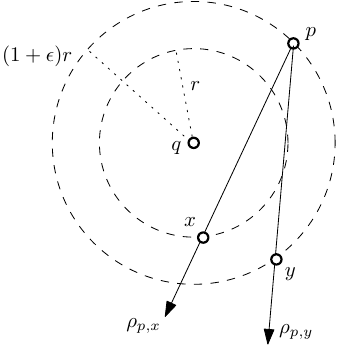}
    \caption{Illustration of Lemma~\ref{lmm:euc:lmm:theta-graph:case1:help}}
    \label{fig:app:euc:lmm:theta-graph:1}
\end{figure}

\begin{lemma} \label{lmm:euc:lmm:theta-graph:case1:help}
	Let $x$ and $y$ be two points that are on the surfaces of $B(q, r)$ and $B(q, (1+\eps)r)$, respectively. If $L_2(p, x) = L_2(p, y)$, then the angle  between the rays $\rho_{p, x}$ and $\rho_{p, y}$ is strictly larger than $\eps/8$.
\end{lemma}

See Figure~\ref{fig:app:euc:lmm:theta-graph:1} for an illustration.

\begin{proof} [Proof of Lemma~\ref{lmm:euc:lmm:theta-graph:case1:help}.]
	Denote by $\gamma$ the angle between $\rho_{p, x}$ and $\rho_{p, y}$.
	It must hold that $\gamma > 0$. Indeed, if $\gamma = 0$, then $x$ and $y$ are on the same line, in which case the condition $L_2(p, x) = L_2(p, y)$ implies $x = y$. This contradicts the fact that $x$ and $y$ are on the surfaces of two different balls.

	\vgap

	Assume, for contradiction, that $\gamma \le \eps/8 \le 1/8$. As $\gamma < \pi/2$, we can use Fact~\ref{fact:app:euc:lmm:theta-graph:fact2} to derive
	\myeqn{
		L_2(x, y)
		&<&
		L_2(p, x)\cdot \tan \gamma \nn \\
		\explain{by triangle inequality}
		&\le& (L_2(p, q) + L_2(q, x)) \cdot \tan \gamma \nn \\
		&=& ((1+\eps) r + r) \cdot \tan \gamma \nn \\
		\textrm{(by Fact~\ref{fact:app:lmm:euc:theta-graph:fact1} and $0 < \gamma \le \eps/8$)}
		&\le&
		((1+\eps) r + r) \cdot (2 \gamma) \nn \\
		&\le&
		(2+\eps) r \cdot (\eps/4) \nn \\
		\explain{as $0 < \eps \le 1$}
		&<&
		3 \eps r / 4 \nn \\
		&<&
		\eps r. \nn
	}

	\vslit

	On the other hand, by the triangle inequality, we have $L_2(x, y) \ge L_2(q, y) - L_2(q, x) = \eps r$, thus giving a contradiction.
\end{proof}


In general, given two different points $p_1, p_2 \in P$, we use the term ``segment $p_1p_2$'' to refer to the line segment connecting them. Define
\myeqn{
	u = \text{the intersection point between the ray $\rho_{p, p_\out}$ and the surface of $B(p, l)$}.  \nn
}
Because (as mentioned) $p_\out$ is in $B(p, l)$, the point $p_\out$ must be on the segment $pu$; see Figure~\ref{fig:app:euc:lmm:theta-graph:case1}.

\vgap

We argue that $u$ must be in the interior of $B(q, (1+\eps) r)$. Once this is done, we know that the entire segment connecting $p$ and $u$ --- except point $p$ --- must be in the interior of $B(q, (1+\eps) r)$. Thus, $p_\out$, which is different from $p$, must be in the interior of $B(q, (1+\eps) r)$, which indicates $L_2(p_\out, q) < (1+\eps)r$, as claimed in \eqref{eqn:app:euc:lmm:theta-graph:goal}.

\vgap

As both $p^*$ and $u$ are on the surface of $B(p, l)$, we must be able to travel on the surface of $B(p, l)$ from $p^*$ to $u$. We will do so on a particular curve --- referred to as the {\em critical curve} --- decided as follows:
\myitems{
	\item For each point $p_\text{seg}$ on segment $p^*u$, shoot a ray from $p$ towards $p_\text{seg}$, and take the point $p_\text{curve}$ at which the ray intersects the surface of $B(p, l)$.

	\vslit

	\item The critical curve is the set of all the $p_\text{curve}$ points produced.
}
See Figure~\ref{fig:app:euc:lmm:theta-graph:case1} for an illustration of the critical curve. As both $p^*$ and $u$ are in cone $C^*$, the angle between the rays $\rho_{p,p^*}$ and $\rho_{p,u}$ is at most $\eps/32$ (because $G$ is an $(\eps/32)$-graph of $P$). For any points $x, y$ on the critical curve, the angle of the rays $\rho_{p,x}$ and $\rho_{p,y}$ can only be smaller and hence is at most $\eps/32$.

\vgap

Assume, for contradiction, that $u$ is not in the interior of $B(q, (1+\eps) r)$. Remember that $p^*$ is in the interior of $B(q, r)$. As we walk from $p^*$ towards $u$ on the critical curve, we must first hit the surface of $B(q, r)$ at some point $x$ and then hit the surface of $B(q, (1+\eps) r)$ at another point $y$. That both $x$ and $y$ are on the curve means that they are both on the surface of $B(p, l)$ and, hence, $L_2(p, x) = L_2(p, y)$. By Lemma~\ref{lmm:euc:lmm:theta-graph:case1:help}, the angle between the rays $\rho_{p, x}$ and $\rho_{p, y}$ is larger than $\eps/8$. This is impossible because as mentioned the angle can be at most $\eps/32$.

\extraspacing{\bf Case 2: $\bm{L_2(p, p_\out) > l}$.} That is, $p_\out$ is outside $B(p, l)$, as illustrated in Figure~\ref{fig:app:euc:lmm:theta-graph:case2}. We will prove later
\myeqn{
	L_2(p^*, p_\out) < \eps r
	\label{eqn:app:euc:lmm:theta-graph:case2:goal}
}
where $p_\out$ is defined in \eqref{eqn:app:euc:lmm:theta-graph:p-out}. The above will give us
\myeqn{
	L_2(p_\out, q) \le L_2(q, p^*) + L_2(p^*, p_\out) < r + \eps r = (1+\eps)r \nn
}
as claimed in \eqref{eqn:app:euc:lmm:theta-graph:goal}.

\vgap

In the rest of our proof, we will fix
\myeqn{
	\gamma = \text{the angular diameter of $C^*$}. \nn
}
Hence, $\gamma \le \eps/32$ (because $G$ is an $(\eps/32)$-graph).
Define:
\myeqn{
	v &=& \text{the intersection point between the ray $\rho_{C^*}$ and the surface of $B(p, l)$}  \nn \\
	\perp^* &=& \text{the projection point of $p^*$ onto ray $\rho_{C^*}$} \nn \\
	\perp_\out &=& \text{the projection point of $p_\out$ onto ray $\rho_{C^*}$} \nn
}
See Figure~\ref{fig:app:euc:lmm:theta-graph:case2} for an illustration.

\begin{figure}
	\centering
	\includegraphics[height=60mm]{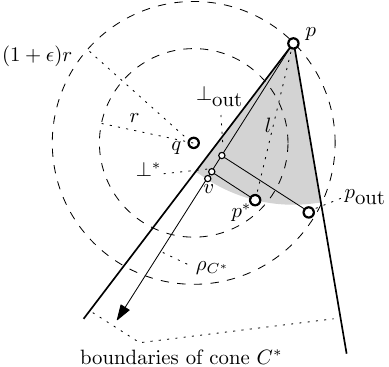}
	\caption{Case 2 of the proof in Section~\ref{app:lmm:euc:theta-graph}}
	\label{fig:app:euc:lmm:theta-graph:case2}
\end{figure}

\begin{lemma} \label{lmm:euc:lmm:theta-graph:case1:help2}
    The following inequalities are correct:
    \myeqn{
		l &<& (2+\eps) \cdot r
		\label{eqn:euc:lmm:theta-graph:case2:0} \\
		L_2(p^*, v) &\le& l \cdot \tan \gamma
		\label{eqn:euc:lmm:theta-graph:case2:1} \\
		L_2(v, \perp_\out)
		&<&
		l \cdot (1 - \cos \gamma) \label{eqn:euc:lmm:theta-graph:case2:3} \\
		L_2(\perp_\out, p_\out) &\le& l \cdot \tan \gamma. \label{eqn:euc:lmm:theta-graph:case2:4}
    }
\end{lemma}

\begin{proof}
	We have
	\myeqn{
		l = L_2(p, p^*) \le L_2(p, q)+L_2(q, p^*)
		< (1+\eps)r + r = (2+\eps)\cdot r \nn
	}
	which proves \eqref{eqn:euc:lmm:theta-graph:case2:0}.

	\vgap

	Next, let us attend to \eqref{eqn:euc:lmm:theta-graph:case2:1}. As both $v$ and $p^*$ are on the surface of $B(p, l)$, it holds that $L_2(p, v) = L_2(p, p^*) = l$. Define $\gamma'$ to be the angle between rays $\rho_{p,v} = \rho_{C^*}$ and $\rho_{p,p^*}$. If $\gamma' = 0$, then $p^*$ coincides with $v$; thus, $L_2(p^*, v) = 0$ and  \eqref{eqn:euc:lmm:theta-graph:case2:1} holds trivially. If $\gamma' > 0$, we must have $\gamma' \le \gamma \le \eps/32$ by definition of $\gamma$. This allows us to apply Fact~\ref{fact:app:euc:lmm:theta-graph:fact2}, which gives:
	\myeqn{
		L_2(p^*, v)
		&<&
		L_2(p, v) \cdot \tan \gamma' \nn \\
		\explain{as $\gamma' \le \gamma \le \eps/32 < \pi/2$}
		&\le&
		l \cdot \tan \gamma \nn
	}
	as claimed in \eqref{eqn:euc:lmm:theta-graph:case2:1}.

	\vgap

	Define $\gamma''$ to be the angle between rays $\rho_{p,v} = \rho_{C^*}$ and $\rho_{p,p_\out}$. We must have $\gamma'' \le \gamma$ by definition of $\gamma$. Thus:
	\myeqn{
		L_2(p, \perp_\out)
		&=&
		L_2(p, p_\out) \cdot \cos \gamma'' \nn \\
		\explain{as $\gamma'' \le \gamma \le \eps/32 < \pi/2$}
		&\ge&
		L_2(p, p_\out) \cdot \cos \gamma \nn \\
		\explain{as $p_\out$ is outside $B(p, l)$}
		&>&
		l \cdot \cos \gamma. \label{eqn:lmm:euc:lmm:theta-graph:case1:help2:1}
	}

	\vslit

	As $p^*$ is on the surface of $B(p, l)$ and the angle between rays $\rho_{p,v}$ and $\rho_{p,p^*}$ is at most $\gamma < \pi/2$, the projection $\perp^*$ of $p^*$ on $\rho_{C^*}$ must be in the interior of $B(p, l)$. Hence, $p^*$ must be on the segment $pv$. On the other hand, by the definition of $p_\out$ (see \eqref{eqn:app:euc:lmm:theta-graph:p-out}), its projection $\perp_\out$ on $\rho_{C^*}$ cannot be farther from $p$ than $\perp^*$. This means that $\perp_\out$ must be on the segment connecting $p$ and $\perp^*$. See Figure~\ref{fig:app:euc:lmm:theta-graph:case2}. Therefore:
	\myeqn{
		L_2(v, \perp_\out)
		&=&
		L_2(p, v) - L_2(p, \perp_\out) \nn \\
		\explain{by \eqref{eqn:lmm:euc:lmm:theta-graph:case1:help2:1}}
		&<&
		l - l \cdot \cos \gamma \nn
	}
	which proves \eqref{eqn:euc:lmm:theta-graph:case2:3}.

	\vgap

	 As mentioned, $\perp_\out$ is on the segment connecting $p$ and $\perp^*$, while $\perp^*$ is on the segment connecting $p$ and $v$. This means that $\perp_\out$ must be on segment $pv$, suggesting $L_2(p, \perp_\out) \le l$. Hence:
	 \myeqn{
		L_2(\perp_\out, p_\out)
		=
		L_2(p, \perp_\out) \cdot \tan \gamma''
		\le
		l \cdot \tan \gamma''
		\le
		l \cdot \tan \gamma
	 }
	 which proves \eqref{eqn:euc:lmm:theta-graph:case2:4}.
\end{proof}

\vgap

Consequently, we have
\myeqn{
	L_2(p^*, p_\out)
	&\le&
	L_2(p^*, v) + L_2(v, p_\out) \nn \\
	&\le&
	L_2(p^*, v) + L_2(v, \perp_\out) + L_2(\perp_\out, p_\out) \nn \\
	\explain{by \eqref{eqn:euc:lmm:theta-graph:case2:1}, \eqref{eqn:euc:lmm:theta-graph:case2:3}, \eqref{eqn:euc:lmm:theta-graph:case2:4}}
	&<&
	l \cdot \tan\gamma + l \cdot (1-\cos \gamma) + l \cdot \tan\gamma \nn \\
	&=&
	l \cdot (2\tan\gamma + 1 - \cos \gamma) \nn \\
	\explain{by \eqref{eqn:euc:lmm:theta-graph:case2:0}}
	&<&
	(2+\eps)(2\tan\gamma + 1 - \cos \gamma) \cdot r
	\nn  \\
	\explain{by Fact~\ref{fact:app:euc:lmm:theta-graph:fact3}}
	&<&
	\eps r \nn
}
as needed in \eqref{eqn:app:euc:lmm:theta-graph:case2:goal}.

\end{document}